\documentclass[11pt]{article}
\usepackage{geometry}
\usepackage{tikz}
\usepackage{url}
\usepackage{amssymb,amsmath}
 \usepackage{pstricks, pst-coil, pst-node, pst-tree, multido}
 \usepackage[english]{babel}
\usepackage{graphics}
\usepackage{graphicx}
 \newtheorem{DE}{Definition}[section]

\newcommand {\sm} {\setminus}
\usepackage{latexsym} %Des symboles de maths en plus.
\usepackage{theorem} %pour changer les styles dans les theoremes.
 \newcommand{\qed}{\relax\ifmmode\hskip2em\Box\else\unskip\nobreak\hfill$\Box$\fi}

\newtheorem{theorem}[DE]{Theorem}
\newtheorem{lemma}[DE]{Lemma}
\newtheorem{conjecture}[DE]{Conjecture}

{\theoremstyle{break}\theorembodyfont{\rmfamily}}
{\theoremstyle{break}\theorembodyfont{\rmfamily}}

\newcounter{claim}
\newenvironment{proof}[1][]%
	{\noindent {\setcounter{claim}{0}\it Proof. }{#1}{}}{\qed\vspace{2ex}}
\newenvironment{claim}[1][]%
	{\refstepcounter{claim}\vspace{1ex}\noindent {(\it\arabic{claim}) {#1}{}}\it}{\vspace{1ex}}
\newenvironment{proofclaim}[1][]%
	{\noindent {}{#1}{}}{ This proves~(\arabic{claim}).\vspace{1ex}}

\newtheorem{innercustomthm}{Theorem}
\newenvironment{customthm}[1]
  {\renewcommand\theinnercustomthm{#1}\innercustomthm}
  {\endinnercustomthm}

\renewcommand{\theenumi}{(\roman{enumi})}
\renewcommand{\labelenumi}{\theenumi
}
\author{Maria Chudnovsky\thanks{Princeton University. This material is
    based upon work supported in part by the U. S. Army Research
    Office under grant number W911NF-16-1-0404, and by NSF grant
    DMS-1763817.  }~, St\'ephan Thomass\'e~\thanks{Univ Lyon, EnsL,
    UCBL, CNRS, LIP, F-69342, LYON Cedex 07, France. Partially
    supported by the LABEX MILYON (ANR-10-LABX-0070) of Universit\'e
    de Lyon, within the program ‘‘Investissements d'Avenir’’
    (ANR-11-IDEX-0007) operated by the French National Research Agency
    (ANR)}~, Nicolas Trotignon\footnotemark[2]~~and Kristina~Vu\v
  skovi\'c\thanks{School of Computing, University of Leeds, UK and
    Faculty of Computer Science (RAF), Union University, Belgrade,
    Serbia.  Partially supported by EPSRC grant EP/N0196660/1, and
    Serbian Ministry of Education and Science projects 174033 and
    III44006.}}

\title{Maximum independent sets in (pyramid, even hole)-free graphs}

\begin{document}

\maketitle

\begin{abstract}
  A \emph{hole} in a graph is an induced cycle with at least 4
  vertices.  A graph is \emph{even-hole-free} if it does not contain
  a hole on an even number of vertices.  A \emph{pyramid} is a graph
  made of three chordless paths $P_1 = a \dots b_1$,
  $P_2 = a \dots b_2$, $P_3 = a \dots b_3$ of length at least~1, two
  of which have length at least 2, vertex-disjoint except at $a$, and
  such that $b_1b_2b_3$ is a triangle and no edges exist between the
  paths except those of the triangle and the three edges incident with
  $a$.

  We give a polynomial time algorithm to compute a maximum weighted
  independent set in a even-hole-free graph that contains no pyramid
  as an induced subgraph.  Our result is based on a decomposition
  theorem and on bounding the number of minimal separators.  All our
  results hold for a slightly larger class of graphs, the class of
  (square, prism, pyramid, theta, even wheel)-free graphs.
\end{abstract}

\section{Introduction}

In this article, graphs are finite and simple.  A \emph{hole} in a
graph is an induced cycle with at least 4 vertices. The \emph{length}
of a hole is the number of vertices in it.  A graph $G$
\emph{contains} a graph $H$ if some induced subgraph of $G$ is
isomorphic to $H$.  A graph $G$ is \emph{$H$-free} if it does not
contain $H$.  When $\cal H$ is a set of graphs, $G$ is
\emph{$\cal H$-free} if it is $H$-free for all $H$ in $\cal H$.

The class of even-hole-free graphs was the object of much research
(see \cite{vuskovic:evensurvey} for a survey). However, the complexity
of computing a maximum independent set in an even-hole-free graph is not
known.

A \emph{pyramid} is a graph made of three chordless paths
$P_1 = a \dots b_1$, $P_2 = a \dots b_2$, $P_3 = a \dots b_3$ of
length at least~1, two of which have length at least 2, vertex-disjoint
except at $a$, and such that $b_1b_2b_3$ is a triangle and no edges
exist between the paths except those of the triangle and the three
edges incident with $a$.  See Fig.~\ref{f:tc}. 

\begin{figure}
  \begin{center}
    \includegraphics[height=2cm]{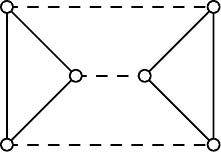}
    \hspace{.2em}
    \includegraphics[height=2cm]{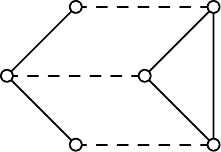} 
    \hspace{.2em}
    \includegraphics[height=2cm]{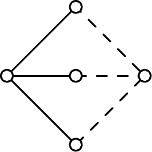}
    \hspace{.2em}
    \includegraphics[height=2cm]{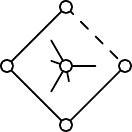}
  \end{center}
  \caption{Prism, pyramid, theta and wheel (dashed lines represent
    paths)\label{f:tc}}
\end{figure}

Our main result is a polynomial time algorithm to compute a maximum
weighted independent set in an (even-hole, pyramid)-free graph.  Our
approach is by first proving a decomposition theorem for the class of
(even-hole, pyramid)-free graph.  This theorem might have other
applications because the presence of a pyramid in an even-hole-free
graphs places significant restrictions on its structure. The graphs
seems to "organize itself" around the pyramid, in a way that can
likely be exploited algorithmically. Results in this direction appear
in Chudnovsky and Seymour~\cite{chudSey:bisimplicial}. So our result
on the pyramid-free case might help to understand the full class of
even-hole-free graphs.  We use our decomposition theorem to prove that
(even-hole, pyramid)-free graph contain polynomially many minimal
separators (to be defined in the next section).  In fact, we prove
this property for a slightly larger class of graphs, namely the
(theta, pyramid, prism, even wheel, square)-free graphs.  And as we
explain in the next section, this property implies the existence of a
polynomial time algorithm to compute maximum weighted independent
sets.

In section~\ref{s:results} we state formally our main results and
motivate them further.  In section~\ref{s:decomp}, we prove the
decomposition theorem.  In section~\ref{s:cutset}, we give several
properties of minimal separators in our class of graphs. In
section~\ref{s:final}, we prove that graphs in our class
contain polynomially many minimal separators.

\subsection*{Notation}

Let $G$ be a graph. By a \emph{path} we mean a chordless (or induced)
path. When $P$ is a path in $G$, we denote by $P^*$ the path induced
by the internal vertices of $P$.  When $a$ and $b$ are vertices of a
path $P$, we denote by $aPb$ the subpath of $P$ with ends $a$ and $b$.
A \emph{clique} in a graph is a set of pairwise adjacent vertices.

When $A, B\subseteq V(G)$, we denote by $N_B(A)$ the
set of vertices of $B\sm A$ that have at least one neighbor in $A$
and $N(A)$ means $N_{V(G)}(A)$.  Note that $N_B(A)$ is disjoint from
$A$.  We write $N(a)$ instead of $N(\{a\})$ and $N[a]$ for $\{a\}\cup
N(a)$.  We denote by $G[A]$ the subgraph of $G$ induced by $A$.  To
avoid too heavy notation, since there is no risk of confusion, when $H$
is an induced subgraph of $G$, we write $N_H$ instead of $N_{V(H)}$. 

A vertex $x$ is \emph{complete} (resp.\ \emph{anticomplete}) to $A$ if
$x\notin A$ and $x$ is adjacent to all vertices of $A$ (resp.\ to no
vertex of $A$). We say that $A$ is \emph{complete} (resp.\
\emph{anticomplete}) to $B$ if every vertex of $A$ is complete
(resp. anticomplete) to $B$ (note that this means in particular that
$A$ and $B$ are disjoint).

\section{Results}
\label{s:results}

Let $G$ be a graph and $a,b \in V(G)$. A set $C\subseteq V(G)$ is an
\emph{minimal $(a,b)$-separator} if $a$ and $b$ are in disctint
components of $G\sm C$ and $C$ is minimal with this property. We say
that $C$ is a \emph{minimal separator} if $C$ is a minimal
$(a,b)$-separator for some pair $a,b$.

It is easy to check that a minimal separator in a graph $G$ can be
equivalently defined as a set $C\subseteq V(G)$ such that $G\sm C$ has
a connected component $L$ and a connected component $R$ such that
every vertex of $C$ has neighbors in both $L$ and $R$.  Note that
$G\sm C$ has possibly more connected components.

Say that a class $\cal C$ of graphs has the \emph{polynomial separator
  property} if there exists $b_{\cal C}$ such that every graph $G$ in
$\cal C$ has at most $|V(G)|^{b_{\cal C}}$ minimal separators.  As
explained by Chudnovsky, Pilipczuk, Pilipczuk and Thomass\'e
in~\cite{DBLP:journals/corr/abs-1903-04761} (see also
the end of Section~\ref{p:compAnalysis}), it follows from results of Bouchit\'e and
Todinca~\cite{DBLP:journals/siamcomp/BouchitteT01,DBLP:journals/tcs/BouchitteT02}
that for any class of graphs, having the polynomial separator property
implies that the Maximum Weighted Independent Set Problem can be
solved in polynomial time.

We are therefore interested in finding classes of graphs where the
number of minimal separators is bounded by some polynomial.  To gain
insight on this question, let us survey examples of graphs with
exponentially many minimal separators.

For an integer $k\geq 1$, the \emph{$k$-prism} is the graph consisting
of two cliques on $k$ vertices, and a $k$-edge matching between
them. More precisely, the $k$-prism $G$ has vertex set
$\{a_1, \dots, a_k, b_1,\dots ,b_k\}$, each of the sets
$\{a_1,\dots ,a_k\}$ and $\{b_1,\dots, b_k\}$ is a clique,
$a_ib_i \in E(G)$ for every $i\in \{1,\dots, k\}$, and there are no
other edges in $G$.  See Fig.~\ref{f:kstuff}. As observed
in~\cite{DBLP:journals/corr/abs-1903-04761}, it is easy to check that
a $k$-prism has $2^k-2$ minimal separators.  This suggests that not
containing a big matching plays a role in bounding the number of
minimal separators, and indeed a simple theorem can be proved in this
direction. Call \emph{$k$-semi-induced matching} any graph whose
vertex set can be partitoned into two sets $X= \{x_1, \dots, x_k\}$ and
$Y= \{y_1, \dots, y_k\}$ such that the only edges between $X$ and $Y$
are the edges $x_iy_i$ ($i=1, \dots, k$).  The edges among vertices of
$X$ and vertices of $Y$ are unrestricted.

\begin{theorem}
  \label{th2}
  For every $k$, every graph $G$ on $n$ vertices that contains no
  $k$-semi-induced matching has at most $O(n^{2k-2})$ minimal
  separators that can be enumerated in time $O(n^{2k})$.
\end{theorem}

\begin{proof}
  Let $a$ and $b$ be two non-adjacent vertices in a graph $G$ that
  does not contain a $k$-semi-induced matching, and let $C$ be a
  minimal separator separating them.  Call $A$ and $B$ the components
  of $G\sm C$ that contain $a$ and $b$ respectively.  By minimality of
  $C$, every vertex in $C$ has a neighbor in $A$.  It is therefore
  well defined to consider an inclusion-wise minimal subset $X_A$ of
  $A$ such that $C\subseteq N(X_A)$.  For every $x\in X_A$, there
  exists a vertex $c\in C$ such that $xc\in E(G)$ and no other vertex
  of $X_A$ is adjacent to $c$.  For otherwise, $X_A\sm \{x\}$ would
  contradict the minimality of $X_A$.  It follows that $G[X_A \cup C]$
  contains an $|X_A|$-semi-induced matching, so $|X_A|< k$.  We may
  define a similar set $X_B$, and we observe that
  $C = N(X_A) \cap N(X_B)$.

  From the previous paragraph, the following algorithm enumerates all
  minimal separators of $G$: for every pair of sets $X_A, X_B$ of
  cardinality less than $k$, compute $C = N(X_A) \cap N(X_B)$ and
  check whether $C$ is a minimal separator.  Since
  ${n\choose i} \leq n^i$, we have
  ${n \choose 0}+\cdots+{n \choose k-1}\leq kn^{k-1}$, so the algorithm
  enumerates at most $O(n^{2k-2})$ minimal separators in time
  $O(n^{2k})$.
 \end{proof}

 Results in other directions can be proved. Chudnovsky, Pilipczuk,
 Pilipczuk, and Thomass\'e~\cite{DBLP:journals/corr/abs-1903-04761}
 proved a graph $G$ that contains no $k$-prism and no hole of length
 at least~5 has at most $|V(G)|^{k+2}$ minimal separators.  But since
 we are interested in even-hole-free graphs, we do not want to exclude
 odd holes.

\begin{figure}
  \begin{center}
  \includegraphics{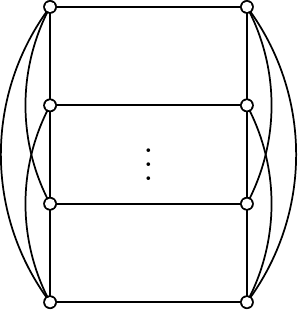}
  \hspace{1.5em}
  \includegraphics{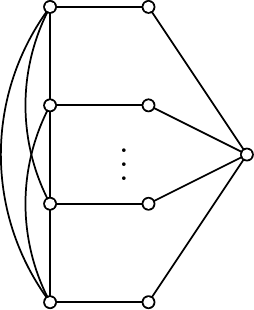}
  \hspace{1.5em}
  \includegraphics{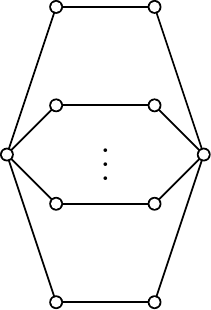}
  \end{center}
  \caption{$k$-prism, $k$-pyramid, $k$-theta\label{f:kstuff}}
\end{figure}

In Fig.~\ref{f:kstuff}, variants of $k$-prims are shown. They are
obtained from $k$-prisms by subdividing the matching edges
(once or twice) and contracting one or two of the cliques into a single
vertex.  We call these graphs $k$-pyramids and $k$-thetas. They are all
easily checked to contain exponentially many minimal separators, and
we do not define them more formally.

From these three examples, we can see that the so-called \emph{3-path
  configurations} are maybe important to understand minimal
separators. They are defined as being the pyramids (that we already
know) and the thetas and prisms that we define now (see
Fig.~\ref{f:tc}). 

A \emph{theta} is a graph made of three internally vertex-disjoint
chordless paths $P_1 = a \dots b$, $P_2 = a \dots b$,
$P_3 = a \dots b$ of length at least~2 and such that no edges exist
between the paths except the three edges incident with $a$ and the three
edges incident with $b$.  

A \emph{prism} is a graph made of three vertex-disjoint chordless paths
$P_1 = a_1 \dots b_1$, $P_2 = a_2 \dots b_2$, $P_3 = a_3 \dots b_3$ of
length at least 1, such that $a_1a_2a_3$ and $b_1b_2b_3$ are triangles
and no edges exist between the paths except those of the two
triangles.

The examples of graphs with exponentially many separators that we have
shown so far all contain a theta, a pyramid or a prism.  But excluding
them is not enough to guaranty a polynomial number of minimum
separators. In Fig.~\ref{f:turtle} a construction of graphs with no
theta, no pyramid and no prism with exponentially many minimal
separator is shown.  We call this construction a \emph{$k$-turtle},
and again it is easy to check that a $k$-turtle contains exponentially
many minimal separators.

\begin{figure}
  \begin{center}
  \includegraphics[height=4cm]{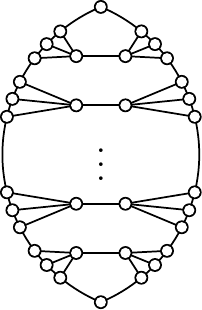}
  \hspace{1.5em}
  \includegraphics[height=4cm]{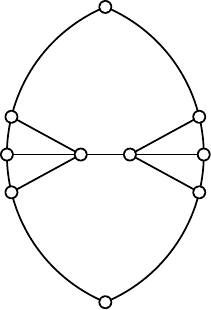}
  \hspace{1.5em}
  \includegraphics[height=4cm]{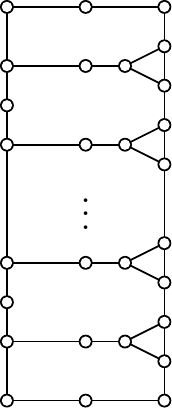}
\end{center}
\caption{$k$-turtle, turtle and $k$-ladder\label{f:turtle}}
\end{figure}

The $k$-turtles suggest defining a \emph{turtle} as a graph made of
two internally vertex-disjoint chordless paths $P_1 = u \dots v$,
$P_2 = u \dots v$ that form a hole. Moreover, there are two adjacent
vertices $x, y$ not in the paths, such that $x$ has at least three
neighbors in $P_1$ (and none in $P_2$) and $y$ has at least three
neighbors in $P_2$ (and none in $P_1$), see Fig.~\ref{f:turtle}.  A
$k$-turtle contains a turtle, in the same way as a $k$-prism contains
a prism, a $k$-pyramid contains a pyramid and a a $k$-theta contains a
theta.  In Fig.~\ref{f:turtle} is also represented a construction that
we call the $k$-ladder, that provides examples of even-hole-free
graphs with maximum degree~3 and exponentially many minimal
separators.  

Since we are not able to imagine examples of graphs with exponentially
many minimal separators containing no prism, pyramid, theta or turtle,
we propose the following conjecture, that would be in some sense the
best possible statement regarding bounding the number of minimal
separators.  

\begin{conjecture}
  \label{conj}
  There is a polynomial $P$ such that every graph $G$ that contains no
  prism, pyramid, theta or turtle has at most $P(|V(G)|)$ minimal
  separators.
\end{conjecture}

Since we are interested in even-hole-free graphs, it is worth
observing that every prism, theta and turtle contains an even hole.
For theta and prism, this is because at least two of the three paths
must have the same parity and therefore form an even hole.  For
turtles, it is because every turtle contains an even wheel. Let us
define them. 

A \emph{wheel} is a graph made of a hole $H$ called the \emph{rim} and
a vertex $v$ called the \emph{center} that has at least three
neighbors in $H$ (see Fig.~\ref{f:tc}).  An \emph{even wheel} is a
wheel whose center has an even number of neighbors in the rim.  It is
easy to check that every turtle contains an even wheel, and that every
even wheel contain an even hole.

A weakening of Conjecture~\ref{conj} is therefore obtained by
restricting it to (prism, pyramid, theta, even wheel)-free graphs.
Note that (prism, theta, even wheel)-free graphs have been studied
under the name of \emph{odd-signable graphs} and they seem to capture
essential properties of even-hole-free graphs, for more about them see
the survey of Vu\v skovi\'c~\cite{vuskovic:evensurvey}.
Interestingly, prisms, pyramids, thetas and wheels are called
\emph{Truemper configurations} and they play an important role in many
decomposition theorems for classes of graphs,
see~\cite{vuskovic:truemper} for a survey.  But we were not able to
prove that (prism, pyramid, theta, even wheel)-free graphs have
polynomially many minimal separators.  However, we can prove that if
we also exclude \emph{squares} (holes of length~4), then the number of
minimal separators is polynomially bounded.

We call $\cal C$ the class of (square, prism, pyramid, theta, even
wheel)-free graphs.  Observe that $\cal C$ is a superclass of the
class of (even hole, pyramid)-free graphs.  Here is our main result
(proved in Section~\ref{s:final}).

\begin{theorem}
  \label{th:main}
  Every graph in $\cal C$ on $n$ vertices contains at most $O(n^{8})$
  minimal separators.  There is an algorithm of complexity~$O(n^{10})$
  that enumerates them. Consequently, there exists a polynomial time
  algorithm for the Maximum Weighted Independent Set restricted to
  $\cal C$.
\end{theorem}

To prove Theorem~\ref{th:main}, we rely on a decomposition theorem for
$\cal C$.  To state it, we need terminology. When $H$ is a hole in
some graph and $u$ is a vertex not in $H$ with at least two neighbors
in $H$, we call \emph{$u$-sector} of $H$ any path of $H$ of length at
least~1, whose ends are adjacent to $u$ and whose internal vertices
are not. Observe that $H$ is edgewise partitioned into its
$u$-sectors.  Let $H$ be a hole in a graph and let $u$ be a vertex not
in $H$. We say that $u$ is \emph{major} w.r.t.\ $H$ if $N_H(u)$ is not
included in a 3-vertex path of $H$.  The decomposition theorem is the
following (proved in Section~\ref{s:decomp}). 

\begin{theorem}
  \label{th:struct}
  Let $G$ be a graph in $\cal C$, $H$ a hole in $G$ and $w$ a major
  vertex w.r.t.\ $H$. If $C$ is a connected component of $G\sm N[w]$,
  then there exists a $w$-sector $P=x\dots y$ of $H$ such that
  $N(C) \subseteq \{x, y\} \cup (N(w) \sm V(H))$.
\end{theorem}

\subsection*{Lower bounds for the number of minimal separators in
  $\cal C$}

For every integer $k$, there exits a graph in $\cal C$ with at least
$O(k^2)$ minimal separators, and this is the best lower bound that we
have so far.  A simple example of this phenomenon is a chordless cycle
of length $k\geq 5$ (any pair of non-adjacent vertices is a minimal
separator). Another example $G_k$ is maybe worth mentioning because
it does not contain holes of length greater than~5.  Let us describe
$G_k$.  Consider four cliques $X$, $Y$, $X'$ and $Y'$, each on $k$
vertices.  Set $X = \{x_1, \dots, x_k\}$, $Y = \{y_1, \dots, y_k\}$,
$X' = \{x'_1, \dots, x'_k\}$, $Y' = \{y'_1, \dots, y'_k\}$.  Add a
vertex $z$. Add all possible edges between $z$ and $X \cup X'$.  Add
all possible edges between $Y$ and $Y'$.  For every $i=1, \dots, k$,
add all possible edges from $x_i$ to
$\{y_{k-i+1}, y_{k-i+2}, \dots, y_k\}$ and all possible edges from
$x'_i$ to $\{y'_{k-i+1}, y'_{k-i+2}, \dots, y'_k\}$.  These are all
the vertices and edges of $G_k$, see Fig.~\ref{f:g4} where $G_4$ is represented.

\begin{figure}
  \begin{center}
  \includegraphics{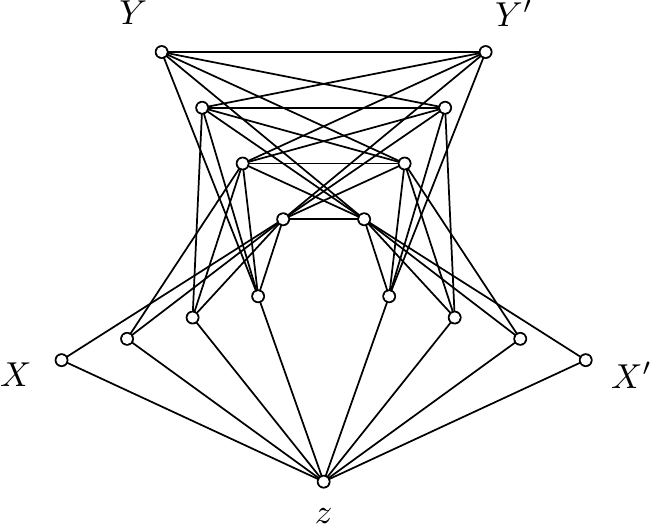}
  \end{center}
  \caption{Graph $G_4$\label{f:g4} (edges of the cliques $X$, $Y$,
  $X'$ and $Y'$ are not represented)}
\end{figure}

It is straightforward to check that $G_k\in \cal C$.  To do so, it is
convenient to note that every hole $H$ in $G_k$ must go through $z$
and contains exactly one vertex in each of the sets $X, X', Y$ and
$Y'$.  So every hole in $G_k$ has length~5.  Since squares, even
wheels, thetas and prisms all contain even holes, the only obstruction
that may exist in $G_k$ is the pyramid. But $G_k$ cannot contain it
since in a pyramid there exists a vertex whose neighborhood contains
three non-adjacent vertices, and this does not exist in $G_k$.

For every $i\in \{1, \dots, k\}$, set
$C_i = \{x_i, x_{i+1}, \dots, x_k\} \cup \{y_{k-i+2}, y_{k-i+3},
\dots, y_k\}$.  Note that for $i\in \{1, \dots, k\}$, $y_1\notin C_i$ and 
$C_1\cap Y = \emptyset$.  For every $j\in \{1, \dots, k\}$, set
$C'_j = \{x'_j, x'_{j+1}, \dots, x'_k\} \cup \{y'_{k-j+2}, y'_{k-j+3},
\dots, y'_k\}$.  It is now easy to check that $C_i\cup C'_j$ is a
minimal separator (separating $z$ from $y_1$) for every pair $(i, j)$ in
$\{1, \dots, k\}^2$.  Since
$|V(G_k)| = 4k+1$ and there are $k^2$ pairs $(i, j)$ in
$\{1, \dots, k\}^2$, $G_k$ has $O(k^2)$ minimal separators.

\subsection*{Rankwidth and semi-induced matchings in $\cal C$}

One may suspect that graphs in $\cal C$ are very ``simple'' in some
sense (in which case our result would be less interesting). And it is
not so easy to exhibit graphs from $\cal C$ that are ``complex'', so
it is worth explaining here how to build such graphs.  To measure the
complexity of a graph we use the notion of rankwidth, that is
equivalent to the notion of cliquewidth, in the sense that a class of
graphs has unbounded rankwidth if and only if it has unbounded
cliquewidth (see~\cite{DBLP:journals/corr/abs-1901-00335} for more
about cliquewidth).

To provide graphs in $\cal C$ of arbitrarily large rankwidth, it
suffices to note that every hole-free graph (better known as
\emph{chordal graphs}) is in $\cal C$. Chordal graphs are known to
have unbounded rankwidth. But chordal graphs are in some sense
``simple'': they are all complete graphs or have clique separators,
and many problems can be solved in polynomial time for them
(see~\cite{vuskovic:truemper} for more about that).

In~\cite{adlerLMRTV:rwehf} Adler, Le, M{\"{u}}ller, Radovanovi\'c,
Trotignon and Vu\v skovi\'c describe even-hole-free graphs of
arbitrarily large rankwidth.  They are also diamond-free (where the
\emph{diamond} is the graph on vertices $a, b, c, d$ with all possible
edges except $ab$) and they have no clique separator.  So, to the best
of our knowledge, they are ``complex''.  The only problem is that they
are not in $\cal C$ because they contain pyramids.  We now explain how
to modify graphs defined in~\cite{adlerLMRTV:rwehf} to obtain graphs
in $\cal C$.

\begin{figure}[htbp]
  \begin{center}  
    \newcommand{\NtoS}[1]{\ifcase #1 % number -> string
          (0)                      \or 0a    \or 0b    %   0,   1,   2,
      \or (1)                      \or 1a              %   3,   4,     
      \or    (1,1)                 \or 11a             %   5,   6,     
      \or         (1,1,1)          \or 111a            %   7,   8,     
      \or                (1,1,1,1) \or 1111a \or 1111b %   9,  10,  11,
      \or                (1,1,1,2) \or 1112a \or 1112b %  12,  13,  14,
      \or                (1,1,1,3) \or 1113a \or 1113b %  15,  16,  17,
      \or                (1,1,1,4) \or 1114a           %  18,  19,     
      \or         (1,1,2)          \or 112a  \or 112b  %  20,  21,  22,
      \or         (1,1,3)          \or 113a            %  23,  24,     
      \or                (1,1,3,1) \or 1131a \or 1131b %  25,  26,  27,
      \or                (1,1,3,2) \or 1132a \or 1132b %  28,  29,  30,
      \or                (1,1,3,3) \or 1133a \or 1133b %  31,  32,  33,
      \or                (1,1,3,4) \or 1134a           %  34,  35,     
      \or         (1,1,4)          \or 114a            %  36,  37,     
      \or    (1,2)                 \or 12a   \or 12b   %  38,  39,  40,
      \or    (1,3)                 \or 13a             %  41,  42,     
      \or         (1,3,1)          \or 131a            %  43,  44,     
      \or                (1,3,1,1) \or 1311a \or 1311b %  45,  46,  47,
      \or                (1,3,1,2) \or 1312a \or 1312b %  48,  49,  50,
      \or                (1,3,1,3) \or 1313a \or 1313b %  51,  52,  53,
      \or                (1,3,1,4) \or 1314a           %  54,  55,     
      \or         (1,3,2)          \or 132a  \or 132b  %  56,  57,  58,
      \or         (1,3,3)          \or 133a            %  59,  60,     
      \or                (1,3,3,1) \or 1331a \or 1331b %  61,  62,  63,
      \or                (1,3,3,2) \or 1332a \or 1332b %  64,  65,  66,
      \or                (1,3,3,3) \or 1333a \or 1333b %  67,  68,  69,
      \or                (1,3,3,4) \or 1334a           %  70,  71,     
      \or         (1,3,4)          \or 134a            %  72,  73,     
      \or    (1,4)                 \or 14a             %  74,  75,     
      \or (2)                      \or 2a    \or 2b    %  76,  77,  78,
      \or (3)                      \or 3a              %  79,  80,     
      \or    (3,1)                 \or 31a             %  81,  82,     
      \or         (3,1,1)          \or 311a            %  83,  84,     
      \or                (3,1,1,1) \or 3111a \or 3111b %  85,  86,  87,
      \or                (3,1,1,2) \or 3112a \or 3112b %  88,  89,  90,
      \or                (3,1,1,3) \or 3113a \or 3113b %  91,  92,  93,
      \or                (3,1,1,4) \or 3114a           %  94,  95,     
      \or         (3,1,2)          \or 312a  \or 312b  %  96,  97,  98,
      \or         (3,1,3)          \or 313a            %  99, 100,     
      \or                (3,1,3,1) \or 3131a \or 3131b % 101, 102, 103,
      \or                (3,1,3,2) \or 3132a \or 3132b % 104, 105, 106,
      \or                (3,1,3,3) \or 3133a \or 3133b % 107, 108, 109,
      \or                (3,1,3,4) \or 3134a           % 110, 111,     
      \or         (3,1,4)          \or 314a            % 112, 113,     
      \or    (3,2)                 \or 32a   \or 32b   % 114, 115, 116,
      \or    (3,3)                 \or 33a             % 117, 118,     
      \or         (3,3,1)          \or 331a            % 119, 120,     
      \or                (3,3,1,1) \or 3311a \or 3311b % 121, 122, 123,
      \or                (3,3,1,2) \or 3312a \or 3312b % 124, 125, 126,
      \or                (3,3,1,3) \or 3313a \or 3313b % 127, 128, 129,
      \or                (3,3,1,4) \or 3314a           % 130, 131,     
      \or         (3,3,2)          \or 332a  \or 332b  % 132, 133, 134,
      \or         (3,3,3)          \or 333a            % 135, 136,     
      \or                (3,3,3,1) \or 3331a \or 3331b % 137, 138, 139,
      \or                (3,3,3,2) \or 3332a \or 3332b % 140, 141, 142,
      \or                (3,3,3,3) \or 3333a \or 3333b % 143, 144, 145,
      \or                (3,3,3,4) \or 3334a           % 146, 147,     
      \or        (3,3,4)           \or 334a            % 148, 149,     
      \or   (3,4)                  \or 34a             % 150, 151,     
      \or (4)                      \or 4a    \or 4b    % 152, 153, 154,
      \or (5)                                          % 155
      \fi}

    \newcommand{\layer}[1]{\ifcase #1 0,155%                    layer 0
      \or 3,76,79,152%                                          layer 1
      \or 5,38,41,74,81,114,117,150%                            layer 2
      \or 7,20,23,36,43,56,59,72,83,96,99,112,119,132,135,148%  layer 3
      \or 9,12,15,18,25,28,31,34,45,48,51,54,61,64,67,70,85,88,91,94,%
          101,104,107,110,121,124,127,130,137,140,143,146%      layer 4
      \or 4,6,8,10,11,13,14,16,17,19,21,22,24,26,27,29,30,32,33,35,37,39,%
          40,42,44,46,47,49,50,52,53,55,57,58,60,62,63,65,66,68,69,71,73,75,%
          77,78,80,82,84,86,87,89,90,92,93,95,97,98,100,102,103,105,106,108,%
          109,111,113,115,116,118,120,122,123,125,126,128,129,131,133,134,%
          136,138,139,141,142,144,145,147,149,151%              layer 5
      \fi}

    \newcommand{\col}[1]{\ifcase #1 %
      black\or blue\or red\or green!70!black\or yellow!60!red\else none\fi}

    \tikzset{v/.style={circle, draw=black, minimum size=1.5mm, inner sep=0pt},
             n/.style={draw=none}
            }

    \begin{tikzpicture}
      \node[v, fill=\col{1}, label=180:$v_1$] (x1) at (195:2.0) {};
      \node[v, fill=\col{2}, label=135:$v_2$] (x2) at (135:2.5) {};
      \node[v, fill=\col{3}, label=105:$v_3$] (x3) at ( 65:2.0) {};
      \node[v, fill=\col{4}, label=  0:$v_4$] (x4) at (320:2.0) {};
      \edef\nlist{\layer{5}}
      \foreach \num in \nlist{
        \pgfmathsetmacro{\angle}{270-360/155*\num}
        \node[v] (\num) at (\angle : 6.0) {};
      }
      \foreach \ind in {4,3,2,1}{ 
        \edef\nlist{\layer{\ind}}
        \pgfmathsetmacro{\dista}{6.3+\ind/7}
        \foreach \num in \nlist{
          \pgfmathsetmacro{\angle}{270-360/155*\num}
          \node[v, fill=\col{\ind}] (\num) at (\angle:6.0) {};
          \draw[\col{\ind}] (x\ind)--(\num);
          \node[n,rotate=\angle] () at (\angle:\dista) {\NtoS{\num}};
        }
      }
      \foreach[count=\mum from 3] \num in {4,5,...,152} \draw (\num)--(\mum);
      \draw (x1)--(x2)--(x3)--(x4)--(x1)--(x3)  (x2)--(x4);
      \node[v, fill=\col{4}, label=  0:$v_4$] (x4) at (320:2.0) {};
    \end{tikzpicture}
  \end{center}
  \caption{A graph from~\cite{adlerLMRTV:rwehf} (the clique contains 4
  vertices)}
  \label{fig:G4}
\end{figure}

Graphs defined in~\cite{adlerLMRTV:rwehf} all vertex-wise partition
into a path $P$ and a clique $K$.  An example is represented in
Fig.~\ref{fig:G4}, where $|K| = 4$ and $P$ is represented as a circle
around $K$.  Every vertex of $K$ has neighbors
in $P$ and every vertex of $P$ has at most one neighbor in $K$.  These
graphs contain pyramids that are built as follows: take three vertices
$a, b, c$ of $K$ that induce a triangle, and consider neighbors
$a', b', c'$ of $a, b, c$ respectively in $P$.  Suppose that $a', b', c'$
are chosen such that $a'Pc'$ has no neighbor of $a$ and $c$ in its
interior and $b'$ is the unique neighbor of $b$ in it.  The pyramid is
then formed by $a'Pc'$ and $abc$. 

To avoid pyramids in graphs from~\cite{adlerLMRTV:rwehf}, take
such a graph $G$ and apply the following algorithm to it:

While there exist a 6-tuple $(a, b, c, a', b', c')$ as above, denote
by $x$ and $y$ the two neighbors of $b'$ in $P$. Remove
$b'$ from $P$, replace it by a path $xp_1p_2p_3p_4p_5p_6p_7y$ and add
the following edges: $bp_1$, $bp_4$ and $bp_7$.  Note that the
obtained graph is still vertex-wise partitioned into a clique and a
path, so that our procedure can be applied repeatedly.

Each time, the number of pyramids in the graph decreases so that the
algorithm terminates.  From the proofs in~\cite{adlerLMRTV:rwehf}, it
is easy to check that we obtain graphs in $\cal C$ that have unbounded
rankwidth. We omit further details that can be found
in~\cite{adlerLMRTV:rwehf}.

We observe that graphs from~\cite{adlerLMRTV:rwehf} may contain
arbitrarily large semi-induced matchings, so that our main result
cannot be a simple corollary of Theorem~\ref{th2}. 

\section{Decomposing graphs in $\cal C$}
\label{s:decomp}

Recall that when $H$ is a hole in some graph and $u$ is a vertex not
in $H$ with at least two neighbors in $H$, we call \emph{$u$-sector}
of $H$ any path of $H$ of length at least~1, whose ends are adjacent
to $u$ and whose internal vertices are not. Observe that $H$ is
edgewise partitioned into its $u$-sectors.

Let $H$ be a hole in a graph and let $u$ be a vertex not in $H$. We
say that $u$ is \emph{major} w.r.t.\ $H$ if $N_H(u)$ is not included
in a 3-vertex path of $H$.  We omit ``w.r.t.\ $H$'' when $H$ is clear
from the context.

\begin{lemma}
  \label{l:vHole}
  In every graph in $\cal C$, every major vertex $u$ w.r.t.\ a hole $H$ has at
  least five neighbors in $H$ or has exactly three neighbors in $H$ that
  are pairwise non-adjacent. 
\end{lemma}

\begin{proof}
  If $u$ has exactly two neighbors in $H$, since these two neighbors are
  not included in a 3-vertex path, they are non-adjacent. Hence, $u$ and
  $H$ form a theta, a contradiction.  If $u$ has exactly three
  neighbors in $H$, since these are not included a 3-vertex path, they
  are pairwise non-adjacent for otherwise $u$ and $H$ form a
  pyramid. If $u$ has exactly four neighbors in $H$, then $u$ and $H$
  form an even wheel, a contradiction.
\end{proof}

It follows from Lemma~\ref{l:vHole} that if $u$ is major w.r.t.~a hole
$H$, then $(H, u)$ is a wheel. We will use this fact throughout the
paper.  A vertex that is not major with respect to some hole $H$ and
still has neighbors in $H$ is \emph{minor} w.r.t.~$H$.

\begin{lemma}
  \label{l:vMinorHole}
  In a graph from $\cal C$, every minor vertex $u$ w.r.t.\ a hole $H$
  satisfies one of the following.
  \begin{itemize}
  \item $u$ has a unique neighbor in $H$ (we then say that $u$ is
    \emph{pending} w.r.t.\ $H$).
  \item $u$ has two neighbors in $H$ which are adjacent (we then say
    that $u$ is a \emph{cap} w.r.t.\ $H$).
  \item $u$ has three neighbors in $H$ which induce a path $xyz$ (we
    then say that $u$ is a \emph{clone of $y$} w.r.t.\ $H$).
  \end{itemize}
\end{lemma}

\begin{proof}
  Otherwise, $u$ has two non-adjacent neighbors in $H$, so $u$ and $H$
  form a theta. 
\end{proof}

When $H$ is hole and $u$ a clone of $y$ w.r.t.\ $H$, we denote by
$H_{u\sm y}$ the hole induced by $\{u\} \cup V(H) \sm \{y\}$. Observe
that $y$ is a clone of $u$ w.r.t.\ $H_{u\sm y}$.

\begin{lemma}
  \label{l:majorClone}
  Let $H$ be a hole in a graph $G \in \cal C$ and $u$ be a clone of
  $y$ w.r.t.~$H$.  Let $v$ be a major vertex w.r.t.\ $H$.  Then,
  $vu\in E(G)$ if and only if $vy\in E(G)$.  In particular, a vertex
  is major w.r.t.\ $H$ if and only if it is major w.r.t.\
  $H_{u\sm y}$.
\end{lemma}

\begin{proof}
  Suppose that $v$ is adjacent to exactly one of $u, y$.  Since $v$ is
  major w.r.t.\ $H$, $(H, v)$ is a wheel. If $(H_{u \sm y}, v)$ is
  also a wheel, then one of $(H, v)$, $(H_{u \sm y}, v)$ is an even
  wheel, a contradiction.  So, $v$ has exactly two neighbors in
  $H_{u \sm y}$, and hence exactly three neighbors in $H$.  By
  Lemma~\ref{l:vHole} the neighbors of $v$ in H are non-adjacent, but
  by Lemma~\ref{l:vMinorHole}, the neighbors of $v$ in $H_{u\sm y}$
  are adjacent, a contradiction.
\end{proof}

\begin{lemma}
  \label{l:vSector}
  Let $u$ and $v$ be two non-adjacent major vertices w.r.t.\ a hole
  $H$ of a graph $G\in \cal C$. Let $P = u'\dots u''$ be a $u$-sector of
  $H$. Then one of the following holds.
  \begin{enumerate}
  \item\label{i:vSector2} $P$ contains at most one neighbor of $v$, and if it has one,
    it is either $u'$ or $u''$.
  \item\label{i:vSector1} $u'u''\in E(G)$ and $v$ is adjacent to both
    $u'$ and $u''$.
  \item\label{i:vSector3} $P$ contains at least 3 neighbors of $v$.
  \end{enumerate}
\end{lemma}

\begin{proof}
  Let $R = x \dots y$ be the path induced by $V(H) \sm V(P)$, with
  ends such that $u'x\in E(G)$ and $u''y\in E(G)$. 

  \begin{claim}
    \label{c:vSector}
    $u$ has a neighbor in the interior of $R$ (in particular, $R$ has
    length at least~2).
  \end{claim}

  \begin{proofclaim}
    Otherwise, $N_H(u) \subseteq \{u', u'', x, y\}$, contradicting
    Lemma~\ref{l:vHole}.
  \end{proofclaim}
  
  Suppose first that $P$ contains exactly one neighbor $v'$ of $v$.
  Suppose for a contradiction that $v'$ is not an end of $P$.  If
  $vx\in E(G)$ then $u'v'\notin E(G)$ because $G$ is
  square-free. Hence, $V(P) \cup \{u, v, x\}$ induces a theta (if
  $ux\notin E(G)$) or a pyramid (if $ux\in E(G)$).  So,
  $vx\notin E(G)$. Symmetrically, $vy\notin E(G)$.  Hence, $G$
  contains a theta from $u$ to $v'$: two paths use vertices of $P$,
  and the third one goes through $v$, some neighbor of $v$ in the
  interior of $R$ (which exists since $v$ is major) and some neighbor
  of $u$ in the interior of $R$, which exists
  by~(\ref{c:vSector}).  So, \ref{i:vSector2} holds.

  Suppose now that $P$ contains exactly two neighbors $v'$ and $v''$
  of $v$. If $u'u''\in E(G)$, then \ref{i:vSector1} holds, so we may
  assume that $u'u''\notin E(G)$.  Hence, $u$ and $P$ form a hole that
  we denote by $H_u$. We have $v'v''\in E(G)$ for otherwise, $v$ and
  $H_u$ form a theta.  By Lemma~\ref{l:vHole}, $v$ has at least five
  neighbors in $H$, so at least one of them is in the interior of $R$.
  Also, $u$ has a neighbor in the interior of $R$
  by~(\ref{c:vSector}).  Hence, $H_u$ together with a shortest path
  from $u$ to $v$ with interior in the interior of $R$ form a pyramid,
  a contradiction.

  Finally, if $P$ contains at least three neighbors of $v$,
  then~\ref{i:vSector3} holds.
\end{proof}

Let $H$ be a hole in a graph and let $u$ and $v$ be two vertice not in
$H$. We say that $u$ and $v$ are \emph{nested} w.r.t.\ $H$ if $H$
contains two distinct vertices $a$ and $b$ such that
one $(a,b)$-path of $H$ contains all neighbors of~$u$ in $H$, and the
other one contains all neighbors of~$v$ in $H$. Observe that $u$ and
$v$ may both be adjacent to $a$ or to $b$. Observe that under the
assumption that $v$ has at least two neighbors in $H$ (so that the
notion of $v$-sector is defined), $u$ and $v$ are nested if and only
if there exists a $v$-sector that contains every neighbor of $u$
in~$H$.  If $u$ is a cap, a pending vertex, or a vertex with no
neighbor in $H$, then it is nested with all other vertices not in
$H$.

\begin{lemma}
  \label{l:sumUpMC}
  Let $H$ be a hole in a graph $G\in \cal C$. If $u$ and $v$ are major
  or clones w.r.t.~$H$ and are nested, then $uv\notin E(G)$.  
\end{lemma}

\begin{proof}
  Otherwise, let $H_u$ be the hole formed by $u$ and the $u$-sector of
  $H$ that contains all neighbors of $v$. Then one of $(H, v)$ or
  $(H_u, v)$ is an even wheel, a contradiction.
\end{proof}

If $u$ and $v$ are two vertices not in $H$ and not nested w.r.t.\ $H$,
then they \emph{cross} on $H$.

\begin{lemma}
  \label{l:notNested}
  Let $H$ be a hole in a graph $G\in \cal C$ and let $u$ and $v$ be
  two vertice not in $H$. If $u$ and $v$ cross, then one the following
  holds.
  \begin{enumerate}
  \item\label{i:cross} $H$ contains four vertices $u'$, $u''$, $v'$
    and $v''$ such that:
    \begin{itemize}
    \item $u'$, $v'$, $u''$ and $v''$ are distinct and appear in this
      order along $H$;
    \item $u', u'' \in N(u)$;
    \item $v', v'' \in N(v)$.
    \end{itemize}
  \item\label{i:corn3} $N_H(u) = N_H(v)$, $N_H(u)$ is an independent set and
     $|N_H(u)| = 3$.
   \item\label{i:twoClones} $N_H(u) = N_H(v)$ and both $u$ and $v$ are
     clones w.r.t.~$H$.
  \end{enumerate}
\end{lemma}

\begin{proof}
  Since a vertex with no neighbor in $H$, a cap or a
  vertex pending w.r.t.\ $H$ is nested with any other vertex
  not in $H$, by Lemma~\ref{l:vMinorHole}, $u$ and $v$ are major or clones
  w.r.t.~$H$.  Hence, consider two non-adjacent neighbors $a, b$ of
  $u$ in $H$. Since $v$ is major or clone, $v$ has a neighbor in the
  interior of one $(a, b)$-path $P_v$ of $H$.  We suppose that $a$,
  $b$ and $P_v$ are chosen subject to these properties
  ($ab\notin E(G)$, $v$ has a neighbor in the interior of $P_v)$ and
  so that $P_v$ is minimal.  If $v$ has neighbors in the
  interior of the other $(a, b)$-path of $H$, then~\ref{i:cross}
  holds.

  Otherwise, $N_H(v) \subseteq V(P_v)$.  If $P_v$ is a $u$-sector,
  then $u$ and $v$ are nested, so suppose that $u$ has a neighbor $u'$
  that is an internal vertex of $P_v$. By the minimality of $P_v$, 
  $N_H(v) \subseteq \{a, b, u'\}$.  Hence, either $v$ is a clone of $u'$
  w.r.t.~$H$, or by Lemma~\ref{l:vHole} applied to $v$,
  $N_H(v) =\{a, b, u'\}$ and $N_H(v)$ is an independent set. So, 
  if $N_H(u) \geq 4$, then~\ref {i:cross} holds, and if $N_H(u) = 3$,
  then~\ref{i:corn3} or~\ref{i:twoClones} holds. 
\end{proof}

\begin{lemma}
  \label{l:corn}
  Let $H$ be a hole in a graph  $G\in \cal C$ and suppose that $u$ and $v$ are
  two major vertices w.r.t.\ $H$. Then $uv\in E(G)$ if and only if
  $u$ and $v$ cross. 
\end{lemma}

\begin{proof}
  If $u$ and $v$ are nested, then $uv\notin E(G)$ by
  Lemma~\ref{l:sumUpMC}.  It remains to prove the converse: if
  $u$ and $v$ cross, then they are adjacent.  So suppose for a
  contradiction that they are not adjacent.

  We apply Lemma~\ref{l:notNested}. Since $u$ and $v$ are
  major,~\ref{i:twoClones} does not hold.  If~\ref{i:corn3} holds,
  then $G$ contains a square, a contradiction.  Hence we may assume
  that~\ref{i:cross} holds: there exist in $H$ two neighbors $u'$,
  $u''$ of $u$ and two neighbors $v'$, $v''$ of $v$ such that $u'$,
  $v'$, $u''$, $v''$ are distinct and appear in this order along $H$.
  We choose them so that the path $P$ from $u'$ to $u''$ in $H\sm v''$
  is minimal.  We now break into two cases.

    \medskip
    \noindent{\bf Case 1:}  $P$ is not a $u$-sector.

    So, let $u'''$ be a neighbor of $u$ in the interior of $P$.  By
    the minimality of $P$, $u'Pu'''$ and $u'''Pu''$ are $u$-sectors
    and have no neighbor of $v$ in their interior, so $v'=u'''$.  Our
    goal in this case is to show the existence of three paths $R_1, R_2$
    and $R_3$ forming a theta from $u$ to~$v$.  We set $R_1 = u v' v$.  

    If $v$ is adjacent to both $u'$ and $u''$, then $\{u, v, u', v'\}$
    induces a square, a contradiction.  So we may assume up to
    symmetry that $v$ is not adjacent to $u''$.  W.l.o.g.\ we may
    assume that $v''$ is such that $Q = v'\dots v''$ is a $v$-sector
    of $H$ (that contains $u''$). By Lemma~\ref{l:vSector}, $Q$
    contains at least three neighbors of $u$. So, there exists a path
    $R_2$ from $u$ to $v$ going through $v''$ and the interior of
    $R_2$ contains no neighbors of $v'$.

    Let $x$ be the neighbor of $v''$ in $H$ that is not in $Q$ and let
    $R$ be the path of $H$ from $x$ to $u'$ that does not contain
    $P$.

    We claim that $v$ has a neighbor in the interior of $R$ (which
    therefore has length at least~2). Otherwise,
    $N_H(v) \subseteq \{u', v', v'', x\}$.  By Lemma~\ref{l:vHole},
    $N_H(v) = \{u', v', v''\}$ and $u'v'\notin E(G)$. So, $\{u, v, u',
    v'\}$ induces a square, a contradiction. 
 
    We claim that $u$ has a neighbor in the interior of $R$. For
    suppose not. Since $G$ contains no even wheel, $u$ has an odd
    number of neighbors in $Q$, and since it also has an odd number
    of neighbors in $H$, $u$ must be adjacent to $x$. Let
    $Q' = y\dots x$ be the $u$ sector of $H$ that contains $v''$.
    Since $G$ is square-free, $v$ is not adjacent to $x$ or
    $y$. Hence, $Q'$ contains a unique neighbor of $v$, that is in its
    interior, a contradiction to Lemma~\ref{l:vSector}.

    Now, by considering a path $R_3$ from $u$ to $v$ with interior in
    the interior of $R$ (which exists from the two claims we just
    proved), we see that $R_1$, $R_2$ and $R_3$ form a theta.

    \medskip
    \noindent{\bf Case 2:} $P$ is a $u$-sector. 

    We apply Lemma~\ref{l:vSector} to $P$ and we observe that
    outcomes~\ref{i:vSector2} and~\ref{i:vSector1} do not hold, so
    outcome~\ref{i:vSector3} holds: $P$ contains at least three
    neighbors of $v$.  It follows that there exist two internally
    vertex disjoint paths $R_1$ and $R_2$, both from $u$ to $v$, with
    interior in $P$ and such that $V(R_1) \cup V(R_2)$ induces a hole.
    Let $x$ be the neighbor of $u'$ in $H$ that is not in $P$, and $y$
    be the neighbor of $u''$ in $H$ that is not in $P$.  Let $R$ be
    the path of $H$ from $x$ to $y$ that does not contain $P$.

    We claim that $u$ has a neighbor in the interior of
    $R$. Otherwise, $N_H(u) \subseteq \{u', u'', x, y\}$,
    contradicting Lemma~\ref{l:vHole}.

    We claim that $v$ has a neighbor in the interior of $R$.
    Otherwise, $v''$ must be one of $x$ or $y$, say $x$ up to
    symmetry. But since there is no even wheel in $G$, $v$ has an odd
    number of neighbors in $P$ and in $H$, so $v$ must be adjacent to
    $y$. Since $G$ is square-free, $v$ cannot be adjacent to both
    $u'$, $u''$, so suppose up to symmetry that is not adjacent to
    $u''$. Hence, $u''$ is the unique neighbor of $u$ in some
    $v$-sector of $H$ (moreover in its interior), a contradiction to
    Lemma~\ref{l:vSector}.

    Now, by considering a path $R_3$ from $u$ to $v$ with interior in
    the interior of $R$ (which exists from the two claims we just
    proved), we see that $R_1$, $R_2$ and $R_3$ form a theta, a
    contradiction.
\end{proof}

\begin{lemma}
  \label{l:sumUp}
  Let $H$ be a hole in a graph $G\in \cal C$. If $u$ and $v$ are
  non-adjacent vertices of $G\sm H$ that cross, then $u$ and $v$ are
  both clones w.r.t.~$H$ and they have exactly two common neighbors on
  $H$.
\end{lemma}

\begin{proof}
  Since a vertex with no neighbor in $H$, a cap or a pendant vertex is
  nested with any other vertex outside $H$, by
  Lemma~\ref{l:vMinorHole}, $u$ and $v$ are major or clones. If they
  are both major, there is a contradiction by Lemma~\ref{l:corn}.  If
  $u$ is a clone of $y$ and $v$ is major (or vice versa), then by
  Lemma~\ref{l:majorClone}, $vy\notin E(G)$, and it follows that $u$
  and $v$ are nested, a contradiction.  If $u$ and $v$ are both
  clones, then they have two or three common neighbors on $H$ (because
  they cross). If they have three common neighbors, then $G$ contains
  a square, a contradiction. Hence, they have two common neighbors as
  claimed.
\end{proof}

\begin{lemma}
  \label{l:zeroNested}
  Let $H$ be a hole in a graph $G \in \cal C$ and let $P = u\dots v$ be a
  path of length at least~1, vertex-disjoint from $H$, and such that
  $u$ and $v$ have neighbors in $H$ and no internal vertex of $P$ has
  neighbor in $H$.  If $u$ and $v$ are nested, then one of the
  following holds (up to a swap of $u$ and $v$):

  \begin{enumerate}
  \item\label{i:zero5} $P$ has length 1, $u$ is major or is a clone,
    and $N_H(v)$ is an edge that contains exactly one neighbor of $u$.
  \item\label{i:zero6} $u$ is a major vertex or a clone, $v$ is a cap
    and $N_H(v) \subseteq N_H(u)$.
  \item\label{i:zero7} $|N_H(v)| = 1$ and $N_H(v) \subseteq N_H(u)$. 
  \item\label{i:zero8} $N_H(u) \cup N_H(v)$ is an edge of $H$.
  \end{enumerate}
\end{lemma}

\begin{proof}
By Lemma~\ref{l:vMinorHole}, $u$ and $v$ are major, clone, cap or
pending. We may therefore consider four cases.  

\medskip 
  \noindent{\bf Case 1.} At least one of $u$ and $v$ is major.

  Up to symmetry, we suppose that $u$ is major.

  Suppose that $v$ is also major. We apply Lemma~\ref{l:corn} to $u$
  and $v$. Since $u$ and $v$ are nested, $P$ has length at least~2.  Hence
  $G$ contains a theta, a contradiction.  So, we may assume that $v$
  is minor.

  Suppose that $v$ is a clone of some vertex $x\in V(H)$.  Since $u$
  and $v$ are nested, $ux\notin E(G)$.  By Lemma~\ref{l:sumUpMC}, $P$
  has length at least~2.  So, $N_H(v)$ is included in some $u$-sector
  $Q$ of $H$ and $P$ and $Q\sm x$ form a theta, a contradiction.
  
  Suppose that $v$ is a cap and $N_H(v) = xy$.  If $x$ and $y$ are in
  the interior of some $u$-sector $Q$ of $H$, then $P$ and $Q$ form a
  pyramid. Hence, there exists a $u$-sector $R=u'\dots u''$ such that
  w.l.o.g.\ $x=u'$ and $y\in V(R)$. If $y = u''$ then~\ref{i:zero6} holds. If
  $y\neq u''$, then $P$ has length~1, for otherwise $P$ and $R$ form a
  pyramid.  Hence,~\ref{i:zero5} holds.

  Suppose that $v$ is pending. Then~\ref{i:zero7} holds for otherwise
  $N_H(v)$ is in the interior of some $u$-sector of $H$ that together
  with $P$ forms a theta.

  \medskip 
  \noindent{\bf Case 2.} None of $u, v$ is major, and at least one of $u, v$ is a
  clone.

  Up to symmetry, suppose that $u$ is a clone of $x$.

  Suppose that $v$ is a clone of $y$. Since $u$ and $v$ are nested, we
  have $x\neq y$ and $xy\notin E(G)$. By Lemma~\ref{l:sumUpMC}, $P$
  has length at least~2, so $P$ and $H\sm \{x, y\}$ form a theta from
  $u$ to $v$, a contradiction.

  Suppose that $v$ is a cap, and let $yz$ be the two neighbors of
  $v$. If $x\in \{y, z\}$, then~\ref{i:zero6} holds, so suppose $x\notin \{y,
  z\}$. Hence $P$ and $H\sm x$ form a pyramid, unless~\ref{i:zero5}
  holds.

   Suppose that $v$ is pending. Then~\ref{i:zero7} holds for otherwise
  $H\sm x$ and $P$ form a theta.
  
\medskip 
    \noindent{\bf Case 3}. None of $u, v$ is major or a clone, and at least one of
  $u, v$ is a cap.

  Up to symmetry, suppose that $u$ is a cap. 

  Suppose that $v$ is also a cap. Then~\ref{i:zero8} holds, for
  otherwise $H$ and $P$ form a prism or an even wheel. 

  Suppose that $v$ is pending. Then~\ref{i:zero7}  holds for otherwise
  $H$ and $P$ form a pyramid.
  
\medskip 
  \noindent{\bf Case 4}. Both $u, v$ are pending vertices.

  Then~\ref{i:zero7} or~\ref{i:zero8} holds, for otherwise  $H$ and
  $P$ form a theta. 

\end{proof}

\begin{lemma}
  \label{l:pathStrong}
  Let $G$ be a graph in $\cal C$, $H$ a hole in $G$ and $w$ a major
  vertex w.r.t.\ $H$. Suppose that $a, w', b, w''$ are four distinct
  vertices of $H$ that appear in this order along $H$ and such that
  $w', w''$ are adjacent to $w$. Then every path $P$ of $G\sm w$ from
  $a$ to $b$ has an internal vertex adjacent to $w$.
\end{lemma}

\begin{proof}
  Consider a counter-example such that $P$ is of minimum length. Note
  that $P$ has length at least~$2$.  Let $H_a$ (resp.~$H_b$) be the
  path of $H$ from $w'$ to $w''$ that contains $a$ (resp.\ $b$).  Let
  $H_{w'}$ (resp.~$H_{w''}$) be the path of $H$ from $a$ to $b$ that
  contains $w'$ (resp.\ $w''$).

  \begin{claim}
    $P^*$ is vertex-disjoint from $V(H)\cup \{w\}$.
  \end{claim}

  \begin{proofclaim}
    Since $P$ is a counterexample, its interior contains no neighbor
    of $w$, and since $a$, $b$, $w'$ and $w''$ are distinct, we have
    $V(P) \cap \{w', w''\} = \emptyset$.  So, an internal vertex of
    $P$ that is in $H$ would yield a smaller counterexample, a
    contradiction to the minimality of $P$. 
  \end{proofclaim}

  We set $Q = P^* = u\dots v$, where $u$ is adjacent to $a$ and $v$ is
  adjacent to $b$ (possibly, $u=v$).

 \begin{claim}
    \label{c:struct1}
    $u$ (resp. $v$) and $w$ are nested. 
  \end{claim}

  \begin{proofclaim}
    Since $P$ is a counterexample, $u$ and $w$ are non-adjacent. Since
    $w$ is major, by Lemma~\ref{l:sumUp}, $u$ and $w$ are
    nested. Similarly, $v$ and $w$ are nested. 
  \end{proofclaim}
  
  \begin{claim}
    \label{c:struct2}
    $u$ and $v$ are distinct and nested w.r.t.~$H$.
  \end{claim}

  \begin{proofclaim}
    By~(\ref{c:struct1}), $N_H(u)\subseteq V(H_a)$ and
    $N_H(v)\subseteq V(H_b)$. So, $u$ and $v$ are distinct (because
    $a\notin N(v)$) and nested.
  \end{proofclaim}

  \begin{claim}
    \label{c:fin}
    We may assume that $H'=auQvbH_{w''}a$ is a hole that contains all
    neighbors of $w$ in $H$ except $w'$.
   
  \end{claim}

  \begin{proofclaim}
    By the minimality of $P$, no internal vertex of $Q$ has a neighbor
    in $H_a^*$ or in $H_b^*$.  It follows that
    $N_H(Q^*) \subseteq \{w', w''\}$.

    Suppose first that $N_H(Q^*) = \{w', w''\}$. Then, $H$ together
    with a path from $w'$ to $w''$ with interior in $Q^*$ form a
    theta from $w'$ to $w''$, a contradiction.

    Suppose now that $N_H(Q^*) = \emptyset$.  Then,
    by~(\ref{c:struct2}), we may apply Lemma~\ref{l:zeroNested} to
    $Q$.  Since $\{a\} \subseteq N_H(u)\subseteq V(H_a)$ and
    $\{b\} \subseteq N_H(v)\subseteq V(H_b)$, \ref{i:zero6},
    \ref{i:zero7} and~\ref{i:zero8} of Lemma~\ref{l:zeroNested} cannot
    hold.  So~\ref{i:zero5} of Lemma~\ref{l:zeroNested} must hold.  Up
    to symmetry, we may therefore assume that $w'$ is the unique
    common neighbor of $u$ and $v$ on $H$ and
    $\{b\} = N_H(v)\sm \{w'\}$.  Since by~(\ref{c:struct1}) $u$ and
    $w$ are nested,  $w'$ is the unique neighbor of $w$ in $H_{w'}$.
    Also, $a$ may be chosen as close a possible to $w''$ along $H_a$,
    so that $H'=auQvbH_{w''}a$ is a hole that contains all
    neighbors of $w$ in $H$ except $w'$.

    Suppose finally that $|N_H(Q^*)| = 1$.  Up to symmetry we may
    assume $N_H(Q^*) = \{w'\}$.  If $w$ has a neighbor $z$ in
    $H_{w'}^*\sm w'$, then suppose up to symmetry that it is in
    $aH_{w'}w'$.  We see that the four vertices $a$, $z$, $w'$ and
    $w''$ appear in this order along $H$, so that a path from $a$ to
    $w'$ with interior in $Q$ contradicts the minimality of $P$.  It
    follows that $w$ has no neighbor in $H_{w'}^*\sm w'$.  We may
    choose $a$ and $b$ closest to $w''$ along $H_a$ and $H_b$
    respectively.  Since by~(\ref{c:struct1}) $u$ and $w$ are nested
    (and $v$ and $w$ are nested), this implies that $H'=auQvbH_{w''}a$
    is a hole that contains all neighbors of $w$ in $H$ except $w'$.
  \end{proofclaim}

  If $w$ has exactly three neighbors in $H$, then by
  Lemma~\ref{l:vHole} they are pairwise non-adjacent and $H'$
  (from~(\ref{c:fin})) and $w$ form a theta, a contradiction. So, $w$
  has at least five neighbors in $H$, so that $(H', w)$ is a
  wheel. But then, one of $(H, w)$ or $(H', w)$ is an even wheel, a
  contradiction.
\end{proof}

We can now prove Theorem~\ref{th:struct} restated below. 

\begin{customthm}{\ref{th:struct}}
  %\label{th:struct}
  Let $G$ be a graph in $\cal C$, $H$ a hole in $G$ and $w$ a major
  vertex w.r.t.\ $H$. If $C$ is a connected component of $G\sm N[w]$,
  then there exists a $w$-sector $P=x\dots y$ of $H$ such that
  $N(C) \subseteq \{x, y\} \cup (N(w) \sm V(H))$.
\end{customthm}

\begin{proof}
  Set $W= N[w] \cap V(H)$ and $Z = N[w] \sm V(H)$.  Clearly,
  $N(C) \subseteq W \cup Z$.  We have to prove that there exists a
  $w$-sector $P=x\dots y$ of $H$ such that
  $N_W(C) \subseteq \{x, y\}$.  Otherwise, we are in one of the
  following cases.

  \noindent{\bf Case~1:}  there exists $a, w', b, w''$ in $W$, appearing in this
  order along $H$, with $a, b \in N_W(C)$.  In this case, a path from
  $a$ to $b$ with interior in $C$ contradicts
  Lemma~\ref{l:pathStrong}.

  \noindent{\bf Case~2:} $|W|=3$ and $N_W(C) = W = \{x, y, z\}$ (and
  by Lemma~\ref{l:vHole}, $x$ $y$ and $z$ are pairwise non-adjacent).
  In this case, suppose first that $C$ contains a vertex $a$ in
  $H\sm W$.  Up to symmetry, we may assume that $a$ is in the
  $w$-sector of $H$ from $x$ to $y$.  But then, $x$, $a$, $y$ and $z$
  contradict Lemma~\ref{l:pathStrong} because $C$ contains the
  interior of a path from
  $a$ to $z$. Hence, $C \cap V(H) = \emptyset$.  If some vertex
  $v$ of $C$ has more than one neighbor in $\{x, y, z\}$, then $w$ and
  $v$ are contained in a square of $G$, a contradiction.  So, every
  vertex of $C$ has at most one neighbor in $W$.  Consider a
  path $P$ with interior $C$ and that is either from $x$ to $y$, from
  $y$ to $z$, or from $z$ to $x$.  Suppose that $P$ has minimum length
  among all such paths.  Up to symmetry, $P= x\dots y$, and by
  minimality, $P$ contains no neighbor of $z$. It follows that $P$ and
  $H$ from a theta from $x$ to $y$.
\end{proof}

\section{Proper separators}
\label{s:cutset}

A separator in a graph is \emph{proper} if it is minimal and not a
clique.  In view of Theorem~\ref{th:main}, we may restrict our
attention to proper separators because it is known that in any graph
$G$ there exists at most $O(|V(G)|)$ minimal clique separators and
that they can be enumrated in time~$O(|V(G)||E(G)|)$,
see~\cite{DBLP:journals/algorithms/BerryPS10} for details. 

Our goal is to prove that a graph in $\cal C$ contains a ``small''
number of proper separators.  This goal is achieved in the next
section. Here we prove a series of technical lemmas telling where precisely
the vertices of a proper separator are. 

If $C$ is a separator of $G$, a connected component $D$ of $G\sm C$ is
{\em full} if every vertex of $C$ has a neighbor in $D$.

\begin{lemma}
  \label{l:2comp} 
  If $C$ is a proper separator of a graph $G \in \cal C$, then $G\sm C$ has
  exactly two full connected components.  
\end{lemma}

\begin{proof}
  Otherwise, let $c_1c_2$ be a non-edge in $C$ and $X$, $Y$, $Z$ be
  full components of $G\sm C$.  There exits a theta from $c_1$ to
  $c_2$, made of three paths with interior in $X$, $Y$ and $Z$
  respectively.  This is a contradiction.
\end{proof}

In what follows, when $C$ is a proper separator, we denote by $L$ and $R$
the two full components of $G\sm C$ that exist by Lemma~\ref{l:2comp}.
We call a $C$-hole any hole $H$ such that $V(H) \cap C = \{c_1, c_2\}$
where $c_1, c_2$ are non-adjacent vertices, one component of
$H\sm \{c_1, c_2\}$ is in $L$ and the other one is in $R$.  We then
say that $H$ is a $(C, c_1, c_2)$-hole.

For a $(C, c_1, c_2)$-hole $H$, we use notation $H_L$ for the path of
$H$ from $c_1$ to $c_2$ with interior in $L$ and $H_R$ for the path of
$H$ from $c_1$ to $c_2$ with interior in $R$.  We let $l_1$ be the
neighbor of $c_1$ in $H_L$.  We define similarly vertices $r_1$, $l_2$,
and $r_2$.

A $C$-hole $H$ is \emph{clean} w.r.t.~$C$ if every major vertex
w.r.t.~$H$ is in $C$.  The next lemma shows that clean holes exist.

\begin{lemma}
  \label{l:existClean}
  Let $C$ be a proper separator of a graph $G\in \cal C$.  If $c_1$ and
  $c_2$ are non-adjacent vertices of $C$, then a shortest 
  $(C, c_1, c_2)$-hole $H$ is clean w.r.t.~$C$
\end{lemma}

\begin{proof}
  Consider a vertex $v\notin C$ that is major w.r.t\ $H$. If
  $N_H(v) \subseteq V(H_L)$, then a shorter $(C, c_1, c_2)$-hole
  exists (using $v$ as a shortcut), a contradiction. Similarly, there
  is a contradiction if $N_H(v) \subseteq V(H_R)$.  It follows that
  $v$ has neighbors in both $H_L^*$ and $H_R^*$, and in particular in
  both $L$ and $R$, a contradiction. This proves that $H$ is clean
  w.r.t.\ $C$.
\end{proof}

Let $C$ be a proper separator of a graph $G$ and $H$ be a $C$-hole. A
vertex in $G$ is \emph{$(C, H)$-heavy} if it is major w.r.t.\ $H$ and
has neighbors in the interiors of both $H_L$ and $H_R$.  Observe that
a $(C, H)$-heavy vertex must be in $C$, because it has neighbors in
both $L$ and $R$.

\begin{lemma}
  \label{l:heavy}
  Let $C$ be a proper separator of a graph $G\in \cal C$.  Let $H$ and
  $H'$ be two $(C, c_1, c_2)$-holes that are clean w.r.t.~$C$.  A
  vertex in $C$ is $(C, H)$-heavy if and only if it is $(C, H')$-heavy.
\end{lemma}

\begin{proof}
  Otherwise, suppose up to symmetry that some vertex $v$ is
  $(C, H)$-heavy and not $(C, H')$-heavy.  Hence, $v$ has a neighbor
  $v_L$ in the interior of $H_L$ and a neighbor $v_R$ in the interior
  of $H_R$.  
  
  \begin{claim}
    \label{c:heavyClone}
    We may assume that $v$ is a clone of $c_1$ w.r.t.\ $H'$.
  \end{claim}

  \begin{proofclaim}
    The vertices $c_1, v_L, c_2, v_R$ are distinct and appear in this
    order along $H$. By Lemma~\ref{l:pathStrong}, the path $H'_L$ has
    an internal vertex adjacent to $v$.  Similarly, $H'_R$ has an
    internal vertex adjacent to $v$.  Since $v$ is not
    $(C, H')$-heavy, the only possibility is that $v$ is a clone of
    $c_1$ or $c_2$ w.r.t.\ $H'$, and up to symmetry, we suppose it is
    a clone of $c_1$.
  \end{proofclaim}

  \begin{claim}
    \label{c:vLIntern}
    We may assume that $v_L$ is an internal vertex of $l_1H_Lc_2$ (in
    particular, $H_L$ has length at least~3).
  \end{claim}

  \begin{proofclaim}
    Since $v$ is not a clone w.r.t.\ $H$,
    $N_H(v)\not\subseteq\{c_1, l_1, r_1\}$.  Since
    by~(\ref{c:heavyClone}) $vc_2\notin E(G)$, $v$ has a neighbor in
    the interior of either $l_1H_Lc_2$ or $r_1H_Rc_2$.  Up to
    symmetry, we may assume that $v$ has a neighbor in the interior of
    $l_1H_Lc_2$.  Hence, $v_L$ can be chosen in the interior of
    $l_1 H_L c_2$.
  \end{proofclaim}

  \begin{claim}
    \label{c:heavyNeq}
    $l_1\neq l'_1$.
  \end{claim}

  \begin{proofclaim}
    Otherwise the vertices $l_1, v_L, c_2, v_R$ are distinct and
    appear in this order along $H$. By Lemma~\ref{l:pathStrong}, the
    path $l_1 H'_L c_2$ has an internal vertex adjacent to $v$, a
    contradiction to~(\ref{c:heavyClone}).  
  \end{proofclaim}
  
  \begin{claim}
    \label{c:heavyNeighlp1}
    $\{c_1\} \subseteq N_H(l'_1)\subseteq \{c_1, l_1\}$.
  \end{claim}
  
  \begin{proofclaim}
    Otherwise $l'_1$ has two non-adjacent neighbors in $H_L$, and
    since $H$ is clean w.r.t.\ $C$, by Lemma~\ref{l:vMinorHole},
    $l'_1$ is a clone of $l_1$ w.r.t.~$H$.  Hence, the hole
    $H_{l'_1\sm l_1}$ contains four distinct vertices (namely $l'_1$,
    $v_L$, $c_2$, $v_R$). By Lemma~\ref{l:pathStrong}, $v$ has a
    neighbor in the interior of $l'_1H'_Lc_2$. This
    contradicts~(\ref{c:heavyClone}).
  \end{proofclaim}
    
  By~(\ref{c:heavyClone}), $v$ is not adjacent to $c_2$.  It follows
  that $c_2$ is an internal vertex of some $v$-sector $Q$ of $H$.  We
  set $Q = q_L \dots q_R$ with $q_L\in L$ and $q_R \in R$.  Note that
  by~(\ref{c:vLIntern}), $q_L\neq l_1$.  Let $x$ be the vertex of
  $H'_L$ with a neighbor in $Q$, closest to $c_1$ along $H'_L$.
  Note that $x$ exists because of $c_2$.  We set $S = l'_1 H'_L x$.

  \begin{claim}
    \label{c:propS1}
    $S$ has length at least~1.
  \end{claim}

  \begin{proofclaim}
    Otherwise $S$ has length zero, so $x=l'_1$ and $l'_1$ has a
    neighbor in $Q$.  This contradicts~(\ref{c:heavyNeighlp1}).
  \end{proofclaim}

  \begin{claim}
    \label{c:propS}
    $S$ is vertex disjoint from $H$ and the only
    edges between $S$ and $H$ are $l'_1c_1$, possibly $l'_1l_1$, and the
    edges between $x$ and $Q$.  
  \end{claim}
  
  \begin{proofclaim}
    By~(\ref{c:heavyNeq}), $l'_1\notin V(H)$ and
    by~(\ref{c:heavyNeighlp1}), the only edges between $l'_1$ and $Q$
    are $l'_1c_1$ and possibly $l'_1l_1$.  Note that by the definition
    of $x$, $S$ is vertex disjoint from $Q$ and $x$ is the only vertex
    of $S$ with neighbors in $Q$.  Suppose that $S$ contains any vertex $b$ of
    $H$ or that there is any edge $ab$ with $a\in V(S)$, $b\in V(H)$ and
    $ab$ is not $l'_1c_1$, $l'_1l_1$ or an edge between $x$ and $Q$.  Then
     consider the four distinct vertices of $H$: $c_2$, $q_L$, $b$
    and $v_R$.   We see that $V(H'_L) \cup \{b\}$ contains a path $P$
    from $c_2$ to $b$. By~(\ref{c:heavyClone}), $P$ contains no
    internal vertex adjacent to $v$.  This contradicts
    Lemma~\ref{l:pathStrong}.
  \end{proofclaim}
  
  By~(\ref{c:propS1}) and~(\ref{c:propS}), $S$ and $H$ contradict
  Lemma~\ref{l:zeroNested}. 
\end{proof}

By Lemma~\ref{l:heavy}, for a vertex not in $C$, being heavy does not
depend on the choice of a particular hole, but only on the choice of
$C$, $c_1$ and $c_2$.  The notion of $(C, c_1, c_2)$-heavy vertex is
therefore relevant: a vertex is \emph{$(C, c_1, c_2)$-heavy} if for
some (or equivalently every) clean $(C, c_1, c_2)$-hole $H$, it is
$(C, H)$-heavy.  

Until the end of the section, we do not recall in the statements of the
lemmas that we deal with a graph $G$ in $\cal C$, a proper separator $C$,
a clean $(C, c_1, c_2)$-hole $H$ with  the following notation: $l_1$ is the
neighbor of $c_1$ in $H_L$ and $l'_1$ is the neighbor of $l_1$ in
$H_L\sm c_1$.  We define similarly vertices $r_1$, $r'_1$, $l_2$,
$l'_2$, $r_2$ and $r'_2$.

For $i=1, 2$, we denote by $L_i$ the set made of $l_i$ and all the
clones of $l_i$ w.r.t.~$H$.  We denote by $R_i$ the set made of $r_i$
and all the clones of $r_i$ w.r.t.~$H$.  We denote by $C_i$ the set of
vertices of $G$ that are not $(C, c_1, c_2)$-heavy and have neighbors
in both $L_i$ and $R_i$ (observe that $c_i \in C_i$).  Note that
possibly $L_1=L_2$ or $R_1=R_2$ (not both since $H$ is not a $C_4$).
Observe that $L_i$ is possibly not included in $L$, because some
vertices of $L_i$ can be in $C$. Similarly, $R_i$ is possibly not
included in $R$.  And $C_i$ is possibly not included in $C$ because
some vertices of $C_i$ might be in $L$ or in $R$ (not in both, because
as we will see, $C_i$ is clique and $L$ is anticomplete to $R$).

\begin{lemma}
  \label{l:LCRdisjoint}
  For $i\in \{1, 2\}$, $L_i$, $R_i$ and $C_i$ are pairwise disjoint cliques. Moreover,
  $L_i$ is anticomplete to $R_i$, and $C_i$ is anticomplete to $H\sm
  \{c_i, l_i, r_i\}$. 
\end{lemma}

\begin{proof}
  We prove the lemma for $i=1$ ($i=2$ is similar). Clearly, $L_1$ and
  $R_1$ are disjoint.  They are cliques for otherwise, $G$ contains a
  square.  By Lemma~\ref{l:sumUpMC}, $L_1$ is anticomplete to~$R_1$.
  It follows that $C_1$ is disjoint from both $L_1$ and $R_1$.

  Let us prove that $C_1$ is anticomplete to $H\sm \{c_1, l_1, r_1\}$.
  Otherwise, let $c\in C_1$ be a vertex with some neighbor in
  $H\sm \{c_1, l_1, r_1\}$.  Note that $c\neq c_1$. Also,
  $cc_1\in E(G)$ for otherwise $G$ contains a square (with $c$, $c_1$,
  and neighbors of $c$ in $L_1$ and $R_1$).  Let $l\in L_1$ and
  $r \in R_1$ be neighbors of $c$ (they exists by definition of
  $C_1$). We see that $c$ is major w.r.t.\ $H'=(H_{l/l_1})_{r/r_1}$.
  Hence, by Lemma~\ref{l:majorClone} (applied twice) $c$ is major
  w.r.t.\ $H$, and therefore $(C,c_1,c_2)$-heavy, a contradiction to
  the definition of $C_1$.

  It remains to prove that $C_1$ is a clique, so suppose for a
  contradiction that $c$ and $c'$ are non-adjacent vertices of $C_1$.
  Let $l, r$ be neighbors of $c$ in $L_1, R_1$ respectively, and
  $l', r'$ be neighbors of $c'$ in $L_1, R_1$ respectively.  If $c$
  and $c'$ have common neighbors in both $L_1$ and $R_1$, then $G$
  contains a square, a contradiction.  Hence, we may assume that
  $l\neq l'$ and that $cl', c'l\notin E(G)$.
  
  If $c$ and $c'$ have a common neighbor $r''\in R_1$, then the paths
  $r''cl$, $r''c'l'$ and $r''r'_1P l'_1$ form a pyramid.  So,
  $r\neq r'$, $cr'$ and $c'r\notin E(G)$.  Hence, the paths $rcl$,
  $rc'l'$ and $P$ form a prism.
\end{proof}

\begin{lemma}
  \label{l:LCRdisjoint}
  For $i\in \{1, 2\}$, $L_i$, $R_i$ and $C_i$ are pairwise disjoint cliques. Moreover,
  $L_i$ is anticomplete to $R_i$, and $C_i$ is anticomplete to $H\sm
  \{c_i, l_i, r_i\}$. 
\end{lemma}

\begin{proof}
  We prove the lemma for $i=1$ ($i=2$ is similar). Clearly, $L_1$ and
  $R_1$ are disjoint.  They are cliques for otherwise, $G$ contains a
  square.  By Lemma~\ref{l:sumUpMC}, $L_1$ is anticomplete to~$R_1$.
  It follows that $C_1$ is disjoint from both $L_1$ and $R_1$.

  Let us prove that $C_1$ is anticomplete to $H\sm \{c_1, l_1, r_1\}$.
  Otherwise, let $c\in C_1$ be a vertex with some neighbor in
  $H\sm \{c_1, l_1, r_1\}$.  Note that $c\neq c_1$. Also,
  $cc_1\in E(G)$ for otherwise $G$ contains a square (with $c$, $c_1$,
  and neighbors of $c$ in $L_1$ and $R_1$).  Let $l\in L_1$ be a
  neighbor of $c$ (it exists by definition of $C_1$).  We set
  $P=r'_1 H_R c_2 H_L l'_1$ (it has length at least~1, and is induced
  by $V(H)\sm \{l_1, c_1, r_1\}$).  So, $c$ has a neighbor in $P$ and
  we let $x$ be the neighbor of $c$ in $P$ closest to $l'_1$ along
  $P$.  We see that if $cl_1\notin E(G)$, then $l\neq l_1$ and the
  hole $cxPl'_1l_1c_1c$ is the rim of an even wheel with center $l$, a
  contradiction.  Hence, $cl_1\in E(G)$, and symmetrically
  $cr_1\in E(G)$.  It follows that $c$ is major w.r.t.\ $H$ and has
  neighbors in both $H_L^*$ and $H_R^*$, a contradiction to the
  definition of $C_1$.
  
  It remains to prove that $C_1$ is a clique, so suppose for a
  contradiction that $c$ and $c'$ are non-adjacent vertices of $C_1$.
  Let $l, r$ be neighbors of $c$ in $L_1, R_1$ respectively, and
  $l', r'$ be neighbors of $c'$ in $L_1, R_1$ respectively.  If $c$
  and $c'$ have common neighbors in both $L_1$ and $R_1$, then $G$
  contains a square, a contradiction.  Hence, we may assume that
  $l\neq l'$ and that $cl', c'l\notin E(G)$.
  
  If $c$ and $c'$ have a common neighbor $r''\in R_1$, then the paths
  $r''cl$, $r''c'l'$ and $r''r'_1P l'_1$ form a pyramid.  So,
  $r\neq r'$, $cr'$ and $c'r\notin E(G)$.  Hence, the paths $rcl$,
  $rc'l'$ and $P$ form a prism.
\end{proof}

For $i\in \{1, 2\}$, an \emph{$(L, i)$-viaduct w.r.t.\ $(C, H)$} is a
path $Q=u_L\dots u_R$ of $G$ such that:

  \begin{enumerate}
  \item $V(Q^*) \cap (V(H)\cup L_i \cup R_i) = \emptyset$;
  \item $V(Q^*)$ is anticomplete to $V(H\sm c_i)$;
  \item $C\cap V(Q) =\{u_L\}$;
  \item $u_R \in R_i\sm C$ (so possibly, $u_R=r_i$);
  \item one of the following holds:
    \begin{itemize}
    \item
      $u_L$ is major w.r.t.~$H$, $N_H(u_L)\subseteq V(H_L)$, $u_Lc_i \in E(G)$; or
    \item 
      $u_L\in L_i\cap C$.
    \end{itemize}
 \end{enumerate}
  
 For $i\in \{1, 2\}$, an \emph{$(R, i)$-viaduct w.r.t.\ $(C, H)$} is a
 path $Q=u_L\dots u_R$ of $G$ such that:

  \begin{enumerate}
  \item $V(Q^*) \cap (V(H) \cup L_i \cup R_i) = \emptyset$;
  \item $V(Q^*)$ is anticomplete to $V(H\sm c_i)$;
  \item $C\cap V(Q) =\{u_R\}$;
  \item $u_L \in L_i\sm C$ (so possibly, $u_L=l_i$);
  \item one of the following holds:
    \begin{itemize}
    \item
      $u_R$ is major w.r.t.~$H$, $N_H(u_R)\subseteq V(H_R)$, $u_Rc_i \in E(G)$; or
    \item 
      $u_R\in R_i\cap C$.
    \end{itemize}
 \end{enumerate}

 We call \emph{viaduct} any path that is an $(L, i)$-viaduct or an
 $(R, i)$-viaduct for $i\in \{1, 2\}$. 

 \begin{lemma}
   \label{l:descViaducts}
   For every $i\in \{1, 2\}$, every $(L,i)$-viaduct and $(R, i)$-viaduct has
   length at least~2 and contains an odd number of neighbors of $c_i$
   (at least~3).
  \end{lemma}

  \begin{proof}
    Suppose $Q = u_L\dots u_R$ is an $(L, 1)$-viaduct (the proof is similar
    for other types of viaducts).

    Since $u_L\in C$ and $u_R\notin C$, we have $u_L \neq u_R$.  So
    $Q$ has length at least~1, and suppose for a contradiction that is
    has length~1.  Then clearly $u_R\neq r_1$ and since $u_L$ and
    $u_R$ are nested, by Lemma~\ref{l:zeroNested} applied to $Q$ and
    $H$, and~\ref{i:zero5} is the only possible outcome.  So, one end
    of $Q$ is a cap, a contradiction to the definition of viaducts
    (observe however that if $u, v$ are vertices like in
    outcome~\ref{i:zero5} Lemma~\ref{l:zeroNested}, then $uvr_1$ is
    possibly a viaduct of length~2).  So $Q$ has length at least~2. 

    Observe that $u_L\notin V(H)$ (while $u_R$ is either in $V(H)$ or is
    a clone of $r_1$).  Let $x$ be the neighbor of $u_L$ in $H_L$,
    closest to $c_2$ along $H_L$.  Consider the hole $J$ induced by
    $V(Q) \cup V(xH_Lc_2) \cup (V(H_R) \sm \{c_1, r_1\})$ (note that
    $r_1$ may be in $J$, when $r_1=u_R$). Now, $c_1\notin V(J)$ and
    $c_1$ contains two non-adjacent neighbors in $J$ (namely $u_L$ and
    $u_R$), hence, by Lemmas~\ref{l:vHole} and \ref{l:vMinorHole}, $c$ is
    a clone or a major vertex w.r.t.~$J$, and it therefore has an odd
    number of neighbors in $J$ (at least~3).
  \end{proof}

  The \emph{potential} of $(C, c_1, c_2)$ is the number of
  $(C, c_1, c_2)$-heavy vertices.  The main result of this section is
  the following.
  
\begin{lemma}
  \label{l:main}
  Let $C$ be a proper separator of a graph $G\in \cal C$.  Let $c_1$
  and $c_2$ be non-adjacent vertices of $C$, chosen such that the
  potential of $(C, c_1, c_2)$ is maximum.  Let $H$ be a clean
  $(C, c_1, c_2)$-hole.  If $c\in C\sm \{c_1, c_2\}$, then one of the
  following statements holds:

  \begin{enumerate}
  \item\label{i:heavy} $c$ is $(C, H)$-heavy;
  \item\label{i:sets} For some $i\in \{1, 2\}$, $c$ has a neighbor in
    $L_i\sm C$ and a neighbor in $R_i\sm C$;
  \item \label{i:viaduct} $c$ is the end of some viaduct
    w.r.t.~$(C, H)$.
  \end{enumerate}
\end{lemma}

\begin{proof}
  Since $C$ is a proper separator and $L$ is connected, there exists a
  path $Q_L = c \dots c_L$ such that $V(Q_L\sm c) \subseteq L$ and
  $c_L$ has neighbors in the interior of $H_L$ (possibly $c = c_L$).
  There exits a similar path $Q_R = c\dots c_R$.  We set
  $Q = c_LQ_L c Q_R c_R$ and suppose that $Q$ is minimal (so $Q_L$ and
  $Q_R$ are shortest paths).

  \begin{claim}
    \label{i:length1}
    We may assume that $Q$ has length at least~1. In particular,
    $c_L$ and $c_R$ are nested w.r.t.\ $H$.
  \end{claim}

  \begin{proofclaim}
    Otherwise, $Q = c= c_L= c_R$. Since $c$ has a neighbor in $H_L^*$
    and in $H_R^*$,  it is either a major vertex or a clone. If it is
    major, then it is heavy w.r.t.\ $H$ and~\ref{i:heavy} holds.  If it
    is a clone, it must be a clone of $c_1$ or $c_2$, so~\ref{i:sets}
    holds.
  \end{proofclaim}
  
  \begin{claim}
    \label{c:intern}
    We may assume that $c_1$ has neighbors in the interior of $Q$ and
    $c_2$ has no neighbors in the interior of $Q$.
  \end{claim}

  \begin{proofclaim}
    Suppose that both $c_1$ and $c_2$ have neighbors in the interior
    of $Q$. Then, $H$ and a shortest path from $c_1$ to $c_2$ with interior in
    the interior of $Q$ form a theta, a contradiction.

    So, suppose that none of $c_1, c_2$ have neighbors in the interior
    of $Q$.  Since by~(\ref{i:length1}) $c_L$ and $c_R$ are nested
    w.r.t.\ $H$, we apply Lemma~\ref{l:zeroNested} to $Q$.  Since
    $c_L$ has neighbors in the interior of $H_L$ and $c_R$ has
    neighbors in the interior of $H_R$, outcomes~\ref{i:zero6},
    \ref{i:zero7} and~\ref{i:zero8} cannot hold.

    Hence outcome \ref{i:zero5} holds.  So $Q=c_Lc_R$,
    $c\in \{c_L, c_R\}$ and exactly one of $c_1$ or $c_2$ (say $c_1$)
    is a common neighbor of $c_L$ and $c_R$, and up to symmetry, $c_L$
    is major or clone of $l_1$, and $c_R$ is a cap.  If $c=c_L$, then
    $c_Lc_Rr_1$ is an $(L, 1)$-viaduct and~\ref{i:viaduct} holds.  If
    $c=c_R$, then since $H$ is clean, $c_L$ cannot be major, so it is
    a clone, $c_L\in L_1\sm C$, $cr_1\in E(G)$ so~\ref{i:sets} holds.

    Hence, we may assume that exactly one of $c_1$ or $c_2$ has
    neighbors in the interior of $Q$, and up to symmetry, we may
    assume that it is $c_1$. 
  \end{proofclaim}

  \begin{claim}
    \label{c:newC}
    If $c_Lc_1\notin E(G)$, then $ N_H(c_L) = \{l_1\}$.  
  \end{claim}

  \begin{proofclaim}
    Let $x_L$ be the neighbor of $c_1$ in $Q$, closest to $c_L$ along
    $Q$ ($x_L$ exists by~(\ref{c:intern}) and $x_L\neq c_L$ by
    assumption).  Since $c_L$ and $x_L$ are nested (because $x_L$ has
    no neighbor in the interior of $H_L$), we may apply
    Lemma~\ref{l:zeroNested} to $c_L Q x_L$.  Since $c_L$ and $x_L$
    have no common neighbor on $H$,~\ref{i:zero5},
    \ref{i:zero6} and~\ref{i:zero7} do not
    hold. Hence~\ref{i:zero8} holds and
    $N_H(c_L) = \{l_1\} $.
  \end{proofclaim}

  \begin{claim}
    \label{c:pendingN}
    If $c_L$ is a cap or a pending vertex, then
    $\{l_1\} \subseteq N_H(c_L) \subseteq \{l_1, c_1\}$.

    If $c_R$ is a cap or a pending vertex, then
    $\{r_1\} \subseteq N_H (c_R) \subseteq \{r_1, c_1\}$.
  \end{claim}
  
  \begin{proofclaim}
    If $c_Lc_1\notin E(G)$, then our claim holds by~(\ref{c:newC}). 
   Otherwise, $c_Lc_1\in E(G)$, $c_L$ must be a cap
    and  $N_H(c_L) = \{c_1, l_1\}$.  The proof is similar for the
    claim about $c_R$.
  \end{proofclaim}

  \begin{claim}
    \label{c:cloneInC}
    If $c_L$  is a clone or a major vertex w.r.t.\ $H$,
    then  $c_1c_L\in E(G)$ and either $c_L\in C$ or $c_L \in L_1 \sm
    C$.

    The analogous statement holds for $c_R$. 
  \end{claim}

  \begin{proofclaim}
    By symmetry, it suffices to prove the statement for $c_L$, so
    assume that $c_L$ is a clone or a major vertex w.r.t.~$H$.
    By~(\ref{c:newC}), $c_Lc_1\in E(G)$.  If $c_L$ is major then
    $c_L\in C$ because $H$ is clean w.r.t.\ $C$.  If $c_L$ is clone and
    $c_L\notin C$, then $c_L\in L_1\sm C$ because $c_L$ has no
    neighbors in $R$.
  \end{proofclaim}

  \begin{claim}
    \label{c:cleftright}
    We may assume that $c_L\in L_1\sm C$ or
    $\{l_1\} \subseteq N_H(c_L) \subseteq \{l_1, c_1\}$, and
    $c_R\in R_1\sm C$ or
    $\{r_1\} \subseteq N_H(c_R) \subseteq \{r_1, c_1\}$.
  \end{claim}

  \begin{proofclaim}
    By symmetry it suffices to prove the statement about $c_L$.  If
    $c_L$ is a pending vertex or a cap, then the result follows
    by~(\ref{c:pendingN}).  So, suppose that $c_L$ is a clone or a
    major vertex. By~(\ref{c:cloneInC}), $c_1c_L\in E(G)$ and either
    $c_L\in C$ or $c_L\in L_1\sm C$.  We may assume that $c_L\in C$,
    and hence $c_R\notin C$. Note that by~(\ref{i:length1}) and since
    $c_1c_L\in E(G)$, if $c_L$ is a clone w.r.t.~$H$, then it is a
    clone of $l_1$, and if it is major then $N_H(c_L) \subseteq
    V(H_L)$.

    If $c_R$ is pending or cap w.r.t.~$H$, then by~(\ref{c:pendingN}),
    $\{r_1\} \subseteq N(c_R) \cap V(H) \subseteq \{r_1, c_1\}$ and
    hence $c_LQc_R r_1$ is an $(L, 1)$-viaduct and~\ref{i:viaduct}
    holds. So, we may assume that $c_R$ is a clone or a major vertex
    w.r.t.~$H$.  By~(\ref{c:cloneInC}) and since $c_R\notin C$, it
    follows that $c_R\in R_1$.  But then $Q$ is an $(L, 1)$-viaduct
    and~\ref{i:viaduct} holds.
  \end{proofclaim}

By~(\ref{c:cleftright}),
$\{l_1\} \subseteq N_H(c_L) \subseteq \{l'_1, l_1, c_1\}$ and
    $\{r_1\} \subseteq N_H(c_R) \subseteq \{r'_1, r_1, c_1\}$.
So, $V(Q) \cup V(H) \sm \{c_1\}$ contains  a hole
 that contains $Q$ and $c_2$, which we denote by $J$.  Note that $J$ is
 a $(C, c, c_2)$-hole.

  \begin{claim}
    \label{i:Jclean}
    $J$ is a clean w.r.t.\ $C$.  
  \end{claim}

  \begin{proofclaim}
    Otherwise, let $d\notin C$ be a vertex that is major w.r.t.~$J$.
    By symmetry, we may assume that $d\in L$.

    Since $d$ is major w.r.t.~$J$ and not major w.r.t.~$H$ (since $H$
    is clean), $d$ must have a neighbor in $Q$.  Let $d_L$
    (resp.~$d_R$) be the neighbor of $d$ in $Q$ that is closest to
    $c_L$ (resp.~$c_R$) along $Q$. Note that $d_L, d_R \in V(c_L Q c)$
    since $d\in L$.  If $d_LQd_R$ is of length greater than~2, then
    $V(Q_L)\cup \{d\}$ contains a path from $c$ to $c_L$ that is
    shorter than $Q_L$, contradicting the minimality of $Q_L$.  So,
    $d_LQd_R$ is of length at most~2.

    Since $d$ is major w.r.t.~$J$, it follows that
    $N_J(d) \not\subseteq V(Q)$. Suppose that
    $N_J(d) \subseteq V(Q) \cup \{c_2\}$. So $d$ is adjacent to $c_2$.
    By Lemma~\ref{l:vHole}, it follows that $d$ has exactly three
    neighbors in $J$ that are furthermore pairwise non-adjacent,
    namely $d_L$, $d_R$ and $c_2$.  But then $d_LQd_R$ and $d$ form a
    square. Therefore $d$ has a neighbor in $H_L\sm \{c_1, c_2\}$.  By
    minimality of $Q_L$, it follows that
    $N_Q(d) \subseteq \{c_L, c'_L\}$, where $c'_L$ is the neighbor of
    $c_L$ in $Q$. By Lemma~\ref{l:vHole} applied to $d$ and $J$, $d$
    has two non-adjacent neighbors in $J\sm Q$.  Since
    $V(J\sm Q) \subseteq V(H)$ and $d$ is not major w.r.t.~$H$, it
    follows that $d$ is a clone of some vertex $d'$ w.r.t.~$H$, where
    $d'$ is an internal vertex of $J\sm Q$ (so
    $d'\notin \{c_1, c_2, l_1\}$ and $d'=l'_1$ is possibly only when
    $c_L$ is not a clone).  So, by Lemma~\ref{l:vHole},
    $N_Q(d) = \{c_L, c'_L\}$.  If $c_L\in L_1$, then
    $(H_{c_L\sm l_1}, d)$ is an even wheel. So
    $N_H(c_L) \subseteq \{l_1, c_1\}$. Note that by minimality of $Q$,
    no internal vertex of $Q$ has a neighbor in $H\sm \{c_1, c_2\}$.
    Let $c'_1$ be the neighbor of $c_1$ in the interior of $Q$ that is
    closest to $c_L$ along $Q$ (it exists by~(\ref{c:intern})). If
    $N_H(c_L) = \{l_1\}$, then $H_{d\sm d'}$ and $c'_LQc'_1$ form a
    theta from $d$ to $c_1$.  So $N_H(c_L) = \{l_1, c_1\}$. Let $D$ be
    the path from $d$ to $l_1$ contained in
    $(H\sm \{d', c_2\})\cup \{d\}$.  Then $D$ and $c_LQc'_1$ form an
    even wheel with center $c_L$.
  \end{proofclaim}

  \begin{claim}
    \label{i:heavyImplies}
    Let $d \in C\sm \{c, c_1, c_2\}$.  If $d$ is $(C, H)$-heavy, then
    $d$ is $(C, J)$-heavy.
  \end{claim}

  \begin{proofclaim}
    For suppose that $d$ is $(C, H)$-heavy but not $(C, J)$-heavy.  So
    $d$ is major w.r.t.~$H$ and has neighbors in both $H_L^*$ and
    $H_R^*$.

    Suppose that $d$ does not have a neighbor in $J_L^*$.  Then
    $c_L\in L_1 \sm C$, $\{l_1\} \subseteq N_{H_L}(d) \subseteq \{l_1,
    c_1\}$, and $d$ is not adjacent to $c_L$.  But then since $d$ is
    major w.r.t.~$H$, by Lemma~\ref{l:vHole}, it follows that
    $H_{c_L\sm l_1}$ and $d$ form either an even wheel with center $d$
    or a theta. So $d$ has a neighbor in $J_L^*$, and by symmetry $d$
    has a neighbor in $J_R^*$.  Since $d$ is not major w.r.t.~$J$, it
    follows that $d$ is a clone of $c$ or $c_2$ w.r.t.~$J$.

    If $d$ is a clone of $c$ w.r.t.~$J$, then $N_H(d)\subseteq \{l_1, c_1,
    r_1\}$ contradicting the assumption that $d$ is major
    w.r.t.~$H$. So $d$ is a clone of $c_2$ w.r.t.~$J$.

    Since $H$ is of length greater than~4, w.l.o.g.\ $H_L$ is of
    length greater than~2. In particular, $J$ contains $l_2$. Since
    $d$ is major w.r.t.~$H$ and $H$ contains $l_2$ and $c_2$, by
    Lemma~\ref{l:vHole}, $d$ has at least five neighbors in $H$.  It
    follows that $d$ is adacent to $l_1$ or $r_1$.  If $d$ is adjacent
    to $l_1$, then $c_L\in L_1\sm C$ and $d$ is not adjacent to
    $c_L$. But then $(H_{c_L\sm l_1}, d)$ is an even wheel. So $d$ is
    not adjacent to $l_1$ and hence $N_H(d) = \{l_2, c_2, r_2, r_1,
    c_1\}$.  In particular, $r_1\neq r_2$, i.e.\ $H_R$ is of length
    greater than~2. But then, we get a contradiction by a symmetric
    argument.  
  \end{proofclaim}

  To conclude the proof, by~(\ref{i:Jclean}) $J$ is clean
  w.r.t.~$C$.  Also, $c_1$ has neighbor in $J_L^*$ and $J_R^*$, and is
  therefore major w.r.t.~$J$ or a clone of $c$.  In this last case,
  $c_L\in L_1\sm C$, $c_R\in R_1\sm C$ and $Q$ has length~2,
  so~\ref{i:sets} holds.   Hence, we may assume that $c_1$ is major
  w.r.t.~$J$.

  Note that $c$ is not major w.r.t.~$H$.  By~(\ref{i:heavyImplies}),
  we see that the number of $(C, J)$-heavy vertices is greater than
  the number of $(C, H)$-heavy vertices. So, by
  Lemma~\ref{l:heavy}, the potential of $(C, c, c_2)$ is greater than
  then potential of $(C, c_1, c_2)$, a contradiction to the choice of
  $c_1$ and $c_2$.
\end{proof}

\begin{lemma}
  \label{l:candidates}
  For  $i\in \{1, 2\}$, there does not exist both an $(L, i)$-viaduct and
  an $(R, i)$-viaduct. In particular, at least one of $L_i, R_i$
  contains no vertex of $C$. 
  \end{lemma}

  \begin{proof}
    Suppose there exists an $(L, 1)$-viaduct $P= u_L\dots u_R$ and an
    $(R, 1)$-viaduct $Q = v_L \dots v_R$ (the case where $i=2$ is
    similar). Then, $u_L\in C$, $u_R\in R_1$, $v_L\in L_1$ and
    $v_R\in C$.  Note that $N_H(u_L)\subseteq V(H_L)$ and
    $N_H(v_R)\subseteq V(H_R)$.

    \begin{claim}
      \label {c:jCycle}
      $V(P\sm u_L)$ and $V(Q\sm v_R)$ are disjoint and anticomplete.
      Moreover, $u_Lv_R\notin E(G)$.
    \end{claim}

    \begin{proofclaim} 
      The first claim is because $L$ and $R$ are connected components
      of $G\sm C$ and $V(P\sm u_L)\subseteq R$ and
      $V(Q\sm v_R)\subseteq L$.

      Since $u_L$ and $v_R$ are nested and
      both major or clones w.r.t.~$H$, $u_Lv_R\notin E(G)$ follows from
      Lemma~\ref{l:sumUpMC}.
    \end{proofclaim}

    Let $x_L$ be the neighbor of $u_L$ in $v_L l'_1H_L c_2$, closest
    to $v_L$ along this path.  Let $y_L$ be the neighbor of $u_L$ in
    $v_L l'_1H_L c_2$, closest to $c_2$ along this path.  Note that
    $x_L$ and $y_L$ exist and are distinct from the definition of
    viaducts and Lemma~\ref{l:descViaducts}. 
    (but possibly, $x_L=v_L$, $y_L=l'_1$ and $x_Ly_L\in E(G)$ when
    $u_L$ is a clone w.r.t.~$H$). Let $x_R$ be the neighbor of $v_R$
    in $u_R r'_1H_R c_2$, closest to $u_R$ along this path. Let $y_R$
    be the neighbor of $v_R$ in $u_R r'_1H_R c_2$, closest to $c_2$
    along this path.

    If $x_L\neq v_L$, we set $S_L = u_L x_L H_L l'_1 v_L$.  If
    $x_L=v_L$ we set $S_L = u_L v_L$.  If $x_R\neq u_R$, we set
    $S_R = v_R x_R H_L r'_1 u_R$.  If $x_R=u_R$ we set $S_R =
    v_Ru_R$.

    By~(\ref {c:jCycle}), $J = u_LS_Lv_LQv_RS_Ru_RPu_L$ is a cycle
    whose only possible chords are edges from $u_L$ to $Q\sm v_L$ and from
    $v_R$ to $P\sm u_R$.  And such chords exist for otherwise,  $J$ is a
    hole and by Lemma~\ref{l:descViaducts}, $c_1$ has an even number
    of neighbors in $J$.

    Up to the symmetry between $P$ and $Q$, we suppose that $v_R$ has
    a neighbor in $P\sm u_R$ and let $p$ be the neighbor of $v_R$ in $P$
    closest to $u_L$ along $P$ (note that by~(\ref {c:jCycle}), $p\neq
    u_L$). 

    If $u_L$ has a neighbor in $Q\sm v_L$, then let $q$ be the
    neighbor of $u_L$ in $Q$ closest to $v_R$ along $Q$ (note that
    by~(\ref {c:jCycle}), $q\neq v_R$).  We see that the three paths
    $v_R p P u_L$, $v_R Q q u_L$ and $v_R y_R H_R c_2 H_L y_L u_L$
    form theta, a contradiction.  Hence, $u_L$ has no neighbor in
    $Q\sm v_L$.

    If $x_Ly_L\notin E(G)$, then the three paths $v_R p P u_L$,
    $v_R Q v_L S_L x_L u_L$ and $v_R y_R H_R c_2 H_L y_L u_L$ form a
    theta.  If $x_Ly_L\in E(G)$, then $x_L=v_L$ and $y_L=l'_1$,
    so the three paths $v_R p P u_L$, $v_R Q v_L$ and
    $v_R y_R H_R c_2 H_L l'_1$ form a pyramid.  In every case, there
    is a contradiction.
  \end{proof}

  \begin{lemma}
    \label{l:Cmonotonous}
    Let $c_1$, $c_2$, and $H$ be as in Lemma~\ref{l:main}. Let $i\in \{1, 2\}$.
    \begin{itemize}
    \item Suppose $R_i\cap C=\emptyset$.  Then there exists a vertex
      $c\in C_i$ such that for all $x\in C_i$, $x\in C$ if and only if
      $N_{L_i}(c) \subseteq N_{L_i}(x)$.
    \item Suppose $L_i\cap C=\emptyset$.  Then there exists a vertex
      $c\in C_i$ such that for all $x\in C_i$, $x\in C$ if and only if
      $N_{R_i}(c) \subseteq N_{R_i}(x)$.
    \end{itemize}
  \end{lemma}

  \begin{proof}
    By symmetry, it is enough to prove the first claim for $i=1$, so
    suppose $R_1\cap C=\emptyset$. Note that $C_1\cap C \neq \emptyset$ since
    $c_1\in C_1$.  Let $c\in C_1\cap C$ be such that $N_{L_1}(c)$ is
    minimal (inclusion wise).  By Lemma~\ref{l:LCRdisjoint},
    $N_H(c) \subseteq \{c_1, l_1, r_1\}$, and hence by
    Lemma~\ref{l:main} applied to $c$, $c$ has a neighbor $l$ in
    $L_1\sm C$ and a neighbor $r$ in $R_1\sm C$. Let $x\in C_1$. 
 
  If $N_{L_1}(c) \subseteq N_{L_1}(x)$, then $x$ is adjacent to
  $l\in L_1\sm C\subseteq L$ and to some vertex of
  $R_1= R_1\sm C \subseteq R$.  Hence, $x\in C$.

  Conversely, suppose that $x\notin C$ and
  $N_{L_1}(c) \not\subseteq N_{L_1}(x)$.  This means that $c$ has a
  neighbor $y$ in $L_1$ that is not adjacent to $x$. If $x$ has a
  neighbor $z$ in $L_1\sm N(c)$, then $xzycx$ is a square by
  Lemma~\ref{l:LCRdisjoint}, a contradiction.  Hence,
  $N_{L_1}(x) \subsetneq N_{L_1}(c)$, contradicting the choice of $c$.
\end{proof}

 \section{The main proof}
 \label{s:final}

 We describe two algorithms ${\cal A}_{L, L}$ and ${\cal A}_{L, R}$
 that enumerate some proper separators of an input graph $G$.  Note
 that these algorithms can be applied to any graph.
 Algorithm~${\cal A}_{L, L}$ is described in Table~\ref{t:algo1}.
 Algorithm~${\cal A}_{L, R}$ is very similar to~${\cal A}_{L, L}$,
 only steps \ref{step:enumHalf}--\ref{step:final} slightly differ (the
 roles of $L_2$ and $R_2$ are swapped).  In Table~\ref{t:algo2} we
 indicate what are these steps.

 \newcommand{\noidd}[1]{\parbox[t]{13.5cm}{#1}}

 \newcommand{\idd}[1]{\rule{.5cm}{0cm}\parbox[t]{13cm}{#1}}
 \newcommand{\iddd}[1]{\rule{1cm}{0cm}\parbox[t]{12.5cm}{#1}}
\newcommand{\idddd}[1]{\rule{1.5cm}{0cm}\parbox[t]{12cm}{#1}}

\renewcommand{\theenumi}{\arabic{enumi}.}
\renewcommand{\labelenumi}{\theenumi}
\begin{table}

\begin{enumerate}
\item \label{step:init}  \noidd{Enumerate all pairs of distinct and non-adjacent vertices
 $(c_1, c_2)$ of $G$.
 
 Set $C=\{c_1, c_2\}$.}

\item\label{step:enum4} \idd{Enumerate all 4-tuple of vertices $(l_1, r_1, l_2, r_2)$ such
  that $\{l_1, c_1, r_1\}$ and $\{l_2, c_2, r_2\}$ both induce a path
  of length~2 and $\{l_1, l_2\}$ is anticomplete to $\{r_1, r_2\}$.
  Note that possibly  $l_1=l_2$ or $r_1=r_2$.} 
  
\item \label{step:heavy}\iddd{Add to $C$ every vertex $v$ such that
    $c_1$ and $c_2$ are in two distinct connected components of
    $G\sm (N[v]\sm \{c_1, c_2\})$.  If
    $\{l_1, r_1, l_2, r_2\} \cap C\neq \emptyset$, discard $C$.}

\item\label{step:computeHole}\iddd{In $G\sm ((C\cup N(c_1) \cup
    N(c_2))\sm \{l_1, r_1, l_2, r_2\})$, compute a shortest
    path $H_L$ from $l_1$ to $l_2$ and a shortest path $H_R$ from
    $r_1$ to $r_2$. Set $H= c_1r_1H_Rr_2c_2l_2H_Ll_1c_1$.  If $H$ is
    not a hole, discard $C$.}

\item\label{step:clean}\iddd{Add to $C$ every vertex that is major w.r.t.~$H$.}
    
\item \iddd{Compute the set $L_1$ of clones of $l_1$ w.r.t.\ $H$ and add
  $l_1$ to $L_1$.  Compute similar sets $R_1$, $L_2$ and $R_2$.}

\item \iddd{Compute the set $C_1$ of vertices that have neighbors in both
  $L_1 $ and $R_1$ and that are not major w.r.t.\ $H$.  Note that
  $c_1\in C_1$. Compute a similar set $C_2$.}

\item\label{step:checkCliques}\iddd{Check that
    $L_1, C_1, R_1, L_2, C_2, R_2$ are disjoint cliques, except that
    possibly exactly one of the equalities $L_1=L_2$ and $R_1=R_2$ holds.
    If this check fails, discard $C$.}

\item\label{step:enumHalf}\iddd{Enumerate all pairs of vertices
    $c'_1\in C_1$, $c'_2\in C_2$.}

\item\label{step:half1}\idddd{Add to $C$ all vertices $x$ from $C_1$ such that
  $N_{L_1}(c'_1) \subseteq N_{L_1}(x)$.}

\item\label{step:half2}\idddd{Add to $C$ all vertices $x$ from $C_2$ such that
  $N_{L_2}(c'_2) \subseteq N_{L_2}(x)$.}

\item\label{step:candidate}\idddd{Add to $C$ every vertex $c\in L_1$
    such that there exists a path $Q = c\dots v$ with the following
    properties: $v\in R_1$, $V(Q)\cap C =\emptyset$,
    $V(Q^*) \cap (V(H)\cup L_i \cup R_i) = \emptyset$ and $Q^*$ is
    anticomplete to $H\sm c_1$.}
\item\label{step:candidate2}\idddd{Add to $C$ every vertex $c\in L_2$
    such that there exists a path $Q = c\dots v$ with the following
    properties: $v\in R_2$, $V(Q)\cap C =\emptyset$,
    $V(Q^*) \cap (V(H)\cup L_2 \cup R_2) = \emptyset$ and $Q^*$ is
    anticomplete to $H\sm c_2$.}

 \item\label{step:final}\idddd{Check whether $C$ is proper separator of $G$. If not, discard
   $C$. Otherwise  return $C$, and go to the next step of enumeration
   (so step~\ref{step:init}, \ref{step:enum4} or \ref{step:enumHalf}).}  
   
\end{enumerate}
\caption{\label{t:algo1}Algorithm ${\cal A}_{L, L}$}
\end{table}

\begin{table}
\begin{enumerate}
\setcounter{enumi}{8}
\item\label{step:enumHalfLR}\iddd{Enumerate all pairs of vertices
    $c'_1\in C_1$, $c'_2\in C_2$.}

\item\label{step:half1LR}\idddd{Add to $C$ all vertices $x$ from $C_1$ such that
  $N_{L_1}(c'_1) \subseteq N_{L_1}(x)$.}

\item\label{step:half2LR}\idddd{Add to $C$ all vertices $x$ from $C_2$ such that
  $N_{R_2}(c'_2) \subseteq N_{R_2}(x)$.}

\item\label{step:candidateLR}\idddd{Add to $C$ every vertex $c\in L_1$
    such that there exists a path $Q = c\dots v$ with the following
    properties: $v\in R_1$, $V(Q)\cap C =\emptyset$,
    $V(Q^*) \cap (V(H)\cup L_i \cup R_i) = \emptyset$ and $Q^*$ is
    anticomplete to $H\sm c_1$.}
\item\label{step:candidate2LR}\idddd{Add to $C$ every vertex $c\in R_2$
    such that there exists a path $Q = c\dots v$ with the following
    properties: $v\in L_2$, $V(Q)\cap C =\emptyset$,
    $V(Q^*) \cap (V(H)\cup L_2 \cup R_2) = \emptyset$ and $Q^*$ is
    anticomplete to $H\sm c_2$.}
 \item\label{step:finalLR}\idddd{Check whether $C$ is proper separator of $G$. If not, discard
   $C$. Otherwise  return $C$, and go to the next step of enumeration
   (so step~\ref{step:init}, \ref{step:enum4} or \ref{step:enumHalf}).}  
\end{enumerate}
\caption{\label{t:algo2}Algorithm ${\cal A}_{L, R}$ (only steps
  \ref{step:enumHalf}--\ref{step:final} are described)}
\end{table}

\begin{lemma}
  Let $G$ be a graph in $\cal C$. If we run the two algorithms
  ${\cal A}_{L, L}$ and ${\cal A}_{L, R}$ on $G$, then the output is the list
  of all proper separators of $G$ and the running time is at most
  $O(|V(G)|^{10})$.
\end{lemma}

\begin{proof}
  Because of step~\ref{step:final}, the algorithm obviously outputs a
  list of proper separators of~$G$.  Conversely, consider a proper
  separator $D$, and let us check that at least one of
  ${\cal A}_{L, L}$ or ${\cal A}_{L, R}$ outputs $D$.

  Let $c_1, c_2$ be vertices in $D$ such that the potential of
  $(D, c_1, c_2)$ is maximum. At some point, in step~\ref{step:init},
  the algorithm considers the pair of vertices $(c_1, c_2)$ and
  correctly puts $c_1$ and $c_2$ in $C$.  At this step, $C\subseteq D$.

  Let $J$ be a shortest $(D, c_1, c_2)$-hole. Note that by
  Lemma~\ref{l:existClean}, $J$ is clean.  We use the notation
  $l_1, r_1, l_2, r_2$ as in section~\ref{s:cutset}.  At some point in
  step~\ref{step:enum4}, the algorithm considers the 4-tuple
  $(l_1, r_1, l_2, r_2)$.

  In step~\ref{step:heavy}, we claim that all $(D, c_1, c_2)$-heavy
  vertices are put in $C$.  Indeed, let $v$ be such a
  $(D, c_1, c_2)$-heavy vertex.  Note that by Lemma~\ref{l:heavy}, $v$
  is $(D, J)$-heavy and hence $c_1$ and $c_2$ do not belong to the
  same $v$-sector of $J$.
  By Theorem~\ref{th:struct} applied to $J$ and $v$,
  vertices $c_1$ and $c_2$ are in different components of
  $G\sm (N[v]\sm \{c_1, c_2\})$.  Note that all $(D, J)$-heavy vertices are
  obviously in $D$, since they have neighbors in both $L$ and
  $R$. Hence, at this step, we have $C\subseteq D$.  Note that if
  $l_1$ say is put in $C$ at this step, then $l_1$ is adjacent to some
  vertex of $R$, a contradiction, so it is correct to discard
  $C$. There is a similar argument for $r_1, l_2, r_2$.
  
  We claim that in step~\ref{step:computeHole} a clean
  $(D, c_1, c_2)$-hole is computed.  
  Indeed, by Lemma~\ref{l:main} applied to $J$, all vertices of $D$ are either heavy
  (and these are already in $C$), or adjacent to $c_1$ or $c_2$ from
  definitions of  viaducts and by Lemma~\ref{l:LCRdisjoint}.  So, when we compute the paths,
  all vertices of $D$ are removed. Note that the paths exist, because
  of $J$.  Hence, $V(H_L^*)\subseteq L$ and $V(H_R^*)\subseteq R$
  showing that $H$ is a hole. In particular, it is correct to discard
  $C$ when $H$ is not a hole.  Clearly, $H$ is a $(D, c_1,
  c_2)$-hole.  Since $H$ and $J$ both go through $c_1$, $l_1$, $r_1$,
  $c_2$, $l_2$, $r_2$ and $J$ is a shortest $(D, c_1,
  c_2)$-hole, the length of $H$ is the same as the length of $J$, and
  hence $H$ is also a shortest $(D, c_1,
  c_2)$-hole. So,  by Lemma~\ref{l:existClean}, $H$ is clean w.r.t.\ $D$.  Note
  that $H$ and $J$ are potentially different holes, but they both go
  through the same vertices $c_1, l_1, r_1, c_2, l_2, r_2$ and they
  have the same heavy vertices by Lemma~\ref{l:heavy}.
  
  Since $H$ is clean, in step~\ref{step:clean} it is correct to put in
  $C$ every vertex that is major w.r.t.~$H$.  We still have
  $C\subseteq D$.

  In step~\ref{step:checkCliques}, we know by
  Lemma~\ref{l:LCRdisjoint} (and since $H$ is not a square) that
  discarding $C$ is correct whenever we have to do so.
  
  Now, by Lemma~\ref{l:candidates} and the symmetry between $L$ and
  $R$ we may assume that $R_1\cap D = \emptyset$.  More precisely, if
  $R_1\cap D \neq \emptyset$ then $L_1\cap D = \emptyset$, so at some
  other step of enumeration, the 4-tuple $(r_1, l_1, r_2, l_2)$ (and
  not $(l_1, r_1, l_2, r_2)$) is considered, so that
  $R_1\cap D = \emptyset$.

  By Lemma~\ref{l:candidates}, we may consider two cases:

  \begin{itemize}
  \item Case 1: $R_1\cap D$ and $R_2\cap D$ are empty (algorithm
    ${\cal A}_{L, L}$).
  \item Case 2: $R_1\cap D$ and $L_2\cap D$ are empty (algorithm
    ${\cal A}_{L, R}$).
  \end{itemize}

  We shall prove that in each case, the algorithm that is indicated
  above outputs $D$. The cases being similar, we just handle Case~1,
  so we suppose that $R_1\cap D$ and $R_2\cap D$ are empty.  Hence, by
  Lemma~\ref{l:Cmonotonous} there exists a vertex $c'_1\in C_1$ such
  that for all $x\in C_1$, $x\in D$ if and only if
  $N_{L_1}(c'_1) \subseteq N_{L_1}(x)$, and there exists a similar vertex
  $c'_2\in C_2$.  At some point, in step~\ref{step:enumHalf}, the
  vertices $c'_1, c'_2$ will be considered. And by
  Lemma~\ref{l:Cmonotonous}, steps~\ref{step:half1}
  and~\ref{step:half2} correctly put in $C$ the sets $C_1\cap D$ and
  $C_2\cap D$.

  At this step of the algorithm, all major vertices
  and vertices of $C_1\cap D$ and $C_2\cap D$ are in $C$.  The only
  vertices in $D\sm C$ are therefore in $L_1$ and $L_2$.  Let $c$ be a
  vertex in $C\cap L_1$.  By Lemma~\ref{l:main}, there exists a
  viaduct with end $c$, so that in step~\ref{step:candidate} a path
  $Q$ is detected and $c$ is correctly added to $C$.  Conversely, if
  some path $Q$ is detected in step~\ref{step:candidate}, then
  $c\in D$. For otherwise, since $R_1\cap D=\emptyset$, $Q$ is a path
  from $L$ to $R$, so it must contain a vertex of $D$.  Since all
  vertices of $D\sm (L_1 \cup L_2)$ are in $C$ and therefore not in
  $Q$, they are not used by $Q$, so there is a contradiction.

  Similarly, in step~\ref{step:candidate2}, $D\cap L_2$ is put in $C$.

  Now, $C=D$. Hence, in step~\ref{step:final}, $C$ is detected as a
  proper separator and the algorithm outputs
  $D$ as claimed.

  \newpage
  \noindent{\bf Complexity analysis}
  \label{p:compAnalysis}
  
  The enumeration of all vertices takes time $O(n^8)$, and for each of
  them, all the computations that we do rely on connectivity checks
  that can be implemented to run in time $O(n^2)$ with BFS.  The total
  running time is therefore $O(n^{10})$.
\end{proof}

We can now prove Theorem~\ref{th:main}, restated below. 

\begin{customthm}{\ref{th:main}}
  Every graph in $\cal C$ on $n$ vertices contains at most $O(n^{8})$ minimal
  separators.  There is an algorithm of complexity~$O(n^{10})$ that
  enumerates them. Consequently, there exists a polynomial time algorithm
  for the Maximum Weighted Independent Set restricted to $\cal C$. 
\end{customthm}

\begin{proof}
  For each of the 8-tuple of vertices that is considered by algorithms
  ${\cal A}_{L, L}$ and ${\cal A}_{L, L}$, each algorithm outputs at
  most one proper separator.  Hence, there is at most $O(n^8)$ such
  separators.  As explained at the beginning of
  Section~\ref{s:cutset}, non-proper minimal separators are all clique
  separators, and there are at most $O(n)$ and they can be enumerated
  in time $O(n^3)$. In Section~\ref{s:results}, it is explained why this
  implies that the Maximum Weighted Independent Set restricted to
  $\cal C$ can be solved in polynomial time.

\end{proof}

\subsection*{Complexity of MWIS in $\cal C$}

We do not recall here the definition of a \emph{potential maximal
  clique}, see~\cite{DBLP:journals/tcs/BouchitteT02}.  A potential
maximal clique in a graph $G$ is a subset of $V(G)$ with special
properties.  We denote by $m$ the number of edges in $G$, by $p$ the
number of potential maximal cliques in $G$ and by $s$ be the number of
minimal separators in $G$.  In~\cite{DBLP:journals/tcs/BouchitteT02},
it is proved that $p\leq O(ns^2 + ns + 1)$ (Proposition 22) and that,
given the list of minimal separators, the potential maximal cliques of
$G$ can be listed in time $O(n^2ms^2)$ (Theorem
23). In~\cite{DBLP:conf/soda/LokshantovVV14}, based
on~\cite{DBLP:conf/stacs/FominV10}, it is proved that, given the list
of potential maximal cliques, the MWIS problem can be solved in time
$O(n^5mp)$ in any graph (Proposition 1).  By Theorem~\ref{th:main},
$s\leq O(n^8)$, so $p\leq O(n^{17})$.  Hence, in $\cal C$, the MWIS
problem can be solved in time $O(n^{24})$.

\section{Acknowledgement}

We thank Tara Abrishami, Marcin Pilipczuk and Paul Seymour for useful
discussions. This work was initiated in the Bellairs Research
Institute of McGill University where the four authors were invited at
the 2019 Barbados Graph Theory Workshop, organized by Sergey Norin,
Paul Seymour and David Wood.

%\bibliographystyle{plain}
%\bibliography{../../Bibliographie/articles} 

\begin{thebibliography}{10}

\bibitem{adlerLMRTV:rwehf}
Isolde Adler, Ngoc{-}Khang Le, Haiko M{\"{u}}ller, Marko Radovanovi\'c, Nicolas
  Trotignon, and Kristina~Vu\v skovi\'c.
\newblock On rank-width of even-hole-free graphs.
\newblock {\em Discrete Mathematics {\&} Theoretical Computer Science}, 19(1),
  2017.

\bibitem{DBLP:journals/algorithms/BerryPS10}
Anne Berry, Romain Pogorelcnik, and Genevi{\`{e}}ve Simonet.
\newblock An introduction to clique minimal separator decomposition.
\newblock {\em Algorithms}, 3(2):197--215, 2010.

\bibitem{DBLP:journals/siamcomp/BouchitteT01}
Vincent Bouchitt{\'{e}} and Ioan Todinca.
\newblock Treewidth and minimum fill-in: Grouping the minimal separators.
\newblock {\em {SIAM} J. Comput.}, 31(1):212--232, 2001.

\bibitem{DBLP:journals/tcs/BouchitteT02}
Vincent Bouchitt{\'{e}} and Ioan Todinca.
\newblock Listing all potential maximal cliques of a graph.
\newblock {\em Theor. Comput. Sci.}, 276(1-2):17--32, 2002.

\bibitem{DBLP:journals/corr/abs-1903-04761}
Maria Chudnovsky, Marcin Pilipczuk, Michal Pilipczuk, and St{\'{e}}phan
  Thomass{\'{e}}.
\newblock On the maximum weight independent set problem in graphs without
  induced cycles of length at least five.
\newblock {\em CoRR}, abs/1903.04761, 2019.

\bibitem{chudSey:bisimplicial}
Maria Chudnovsky and Paul Seymour.
\newblock Even-hole-free graphs still have bisimplicial vertices.
\newblock {\em CoRR}, abs/1909.10967, 2019.

\bibitem{DBLP:journals/corr/abs-1901-00335}
Konrad~K. Dabrowski, Matthew Johnson, and Dani{\"{e}}l Paulusma.
\newblock Clique-width for hereditary graph classes.
\newblock {\em CoRR}, abs/1901.00335, 2019.

\bibitem{DBLP:conf/stacs/FominV10}
Fedor~V. Fomin and Yngve Villanger.
\newblock Finding induced subgraphs via minimal triangulations.
\newblock In Jean{-}Yves Marion and Thomas Schwentick, editors, {\em 27th
  International Symposium on Theoretical Aspects of Computer Science, {STACS}
  2010, March 4-6, 2010, Nancy, France}, volume~5 of {\em LIPIcs}, pages
  383--394. Schloss Dagstuhl - Leibniz-Zentrum fuer Informatik, 2010.

\bibitem{DBLP:conf/soda/LokshantovVV14}
Daniel Lokshantov, Martin Vatshelle, and Yngve Villanger.
\newblock Independent set in \emph{P}\({}_{\mbox{5}}\)-free graphs in
  polynomial time.
\newblock In Chandra Chekuri, editor, {\em Proceedings of the Twenty-Fifth
  Annual {ACM-SIAM} Symposium on Discrete Algorithms, {SODA} 2014, Portland,
  Oregon, USA, January 5-7, 2014}, pages 570--581. {SIAM}, 2014.

\bibitem{vuskovic:evensurvey}
Kristina Vu{\v s}kovi{\'c}.
\newblock Even-hole-free graphs: a survey.
\newblock {\em Applicable Analysis and Discrete Mathematics}, 10(2):219--240,
  2010.

\bibitem{vuskovic:truemper}
Kristina Vu{\v s}kovi{\'c}.
\newblock The world of hereditary graph classes viewed through {T}ruemper
  configurations.
\newblock In S.~Gerke S.R.~Blackburn and M.~Wildon, editors, {\em Surveys in
  Combinatorics, London Mathematical Society Lecture Note Series}, volume 409,
  pages 265--325. Cambridge University Press, 2013.

\end{thebibliography}
%\end{document}

\end{document}